\documentclass[12pt]{article}

\usepackage[utf8]{inputenc}

\usepackage[margin=1in]{geometry}
\usepackage{setspace}

\usepackage{amsmath, amsthm, amssymb, amsfonts}
\usepackage{bbm}
\usepackage{bigints}
\usepackage{changepage}
\usepackage{rotating}

\usepackage{graphicx}
\graphicspath{{Pictures/}}
\usepackage{caption}
\usepackage{subcaption}
\usepackage{float}

\usepackage{booktabs}
\usepackage{multirow}

\usepackage{tikz}
\makeatletter
\IfFileExists{tikzlibrarybayesnet.code.tex}{\usetikzlibrary{bayesnet}}{}
\makeatother

\usepackage[toc,page]{appendix}

\usepackage{natbib}
\usepackage{url}

\usepackage[ruled,vlined,linesnumbered]{algorithm2e}
\usepackage{algorithmicx}

\usepackage{xcolor}
\usepackage{enumitem}
\usepackage{enumerate}
\usepackage{authblk}
\usepackage{framed}
\usepackage[normalem]{ulem}
\usepackage{titlesec}
\usepackage{comment}
\usepackage{pdfpages}
\usepackage{epstopdf}

\usepackage[colorlinks,citecolor=blue,urlcolor=blue]{hyperref}

\newtheorem{theorem}{\textbf{Theorem}}

\newtheorem{proposition}{Proposition}

\def\bm{\boldsymbol}
\newcommand{\partref}[1]{Part~\ref{#1}}

\usepackage{amsmath}
\usepackage{amsthm}
\usepackage{amssymb}
\usepackage{graphicx}
\usepackage{color}
\usepackage{enumerate}
\usepackage{natbib}
\usepackage{url} 
\usepackage{multirow}
\usepackage{lineno}
\usepackage{makecell}
\usepackage{hyperref}

\def\0{\boldsymbol 0}

 \title{\bf A Frequency-Domain Approach for Integrating Multiple Functional Time Series}
  \author{Zerui Guo \\
    School of Mathematics, Sun Yat-sen University\\
    Jianbin Tan\\
    Department of Biostatistics and Bioinformatics, Duke University\\
    and \\
    Hui Huang \\
    Center for Applied Statistics and School of Statistics, Renmin University of China}

\date{}

\begin{document}
\maketitle

\begin{abstract}
Integrative analysis of multivariate functional time series (MFTS) is both critical and challenging across many scientific domains. Such data often exhibit complex multi-way dependencies arising from within-curve structures, temporal correlations across curves, and cross-subject interactions, underscoring the need for efficient methods that can jointly capture these dependencies and support accurate downstream analyses. 
In this work, we propose a novel frequency-domain framework based on a marginal dynamic Karhunen--Lo\`eve expansion. The key idea is to integrate individual spectral densities of the MFTS to construct a marginal spectral operator, whose eigenfunctions yield optimal functional filters. These filters transform complex functional observations into a structured multivariate time series representation, providing a powerful foundation for joint modeling and estimation. 
Through extensive simulation studies, we demonstrate the superior performance of the proposed approach. We further validate its practical utility through an application to the imputation and forecasting of air pollutant concentration trajectories in China.
\end{abstract}


\newpage
\clearpage
\setstretch{2} 
\section{Introduction}\label{section_intro}

Recent advances in data collection techniques have greatly expanded the availability of functional time series collected from multiple subjects, such as environmental PM$_{2.5}$ concentration curves, energy consumption readings, age-specific mortality data, and traffic flow profiles \citep{tang2022clustering, ma2024network, tan2024graphical, chang2024modelling}. 
These datasets, commonly referred to as multivariate functional time series (MFTS), often exhibit complex multi-way dependencies, including within-curve structures, temporal correlations across curves, and cross-subject dependencies \citep{guo2023consistency, guo2025factor}. 
The coexistence of these dependencies can substantially undermine the statistical efficiency and accuracy of downstream analyses for MFTS, such as missing-data imputation and forecasting \citep{rubin2020sparsely, guo2024unified}. 
This motivates the development of an integrative analysis framework that can flexibly accommodate multilevel dependence structures in MFTS, improving estimation efficiency and enhancing the reliability of functional forecasting.

Conventional methods for MFTS primarily focus on integrating temporal dependencies, that is, constructing subject-specific representations based on each series’ observations and performing estimation within each series using these individual representations.
Two major paradigms are commonly used: time-domain approaches, which integrate information from the lag-$h$ ($h \geq 0$) covariance structures of individual series \citep{bathia2010identifying, gao2019high}, and frequency-domain approaches, which are built upon spectral density functions derived from these lag covariances \citep{hormann2015dynamic, hallin2018optimal}. 
The time-domain perspective provides a direct framework for modeling autocorrelation structures, and the frequency-domain perspective offers a complementary view by decomposing frequency variation into dynamic components, yielding theoretically optimal representations for dynamic functional data \citep{hormann2015dynamic, tan2024functional, guo2024unified}.
However, these approaches primarily rely on individual-specific models for data representation and therefore neglect cross-individual dependencies in MFTS. This omission may lead to inefficient estimation and reduced accuracy in functional data analysis tasks.

Alternatively, marginal principal component analysis (MPCA) provides an effective approach for integrating information across individuals, enabling tractable analysis of multivariate functional data for downstream statistical tasks \citep{park2015longitudinal, chen2017modelling, koner2024profit}. 
Specifically, these methods aggregate individual covariance functions to derive a set of eigenbases that capture marginal functional patterns shared across subjects. 
Based on these bases, each subject’s data are then modeled using random scores combined with the shared marginal structure. 
Compared with individual-based models, the marginalization approach integrates information from all subjects when estimating the marginal bases, potentially improving statistical efficiency in functional data modeling. 
However, existing marginalization methods mainly account for cross-subject dependence while overlooking temporal dependence within each series, which limits their ability to capture dynamic structures and to produce accurate forecasts.

To address these limitations, we propose a novel frequency-domain integrative framework, termed Spectral Marginal Principal Component Analysis (Spectral MPCA), that performs integrated modeling across both the individual and temporal dimensions of MFTS. 
Specifically, our framework first integrates observations within each subject through lag-$h$ covariance functions to capture within- and cross-curve dependencies. 
Second, it transforms the individual lag-covariance functions into the frequency domain to construct subject-specific spectral density kernels, effectively aggregating temporal dependencies in the observed data. 
Third, it integrates cross-subject interactions through the marginalization of these spectral density kernels, yielding an efficient marginal dynamic Karhunen--Lo\`eve (KL) expansion that accounts for within-curve, cross-curve, and cross-individual dependencies. 
We establish a theoretically grounded foundation for this integrative framework and demonstrate its effectiveness through extensive simulation studies. Finally, we illustrate the superiority of our method for missing-data imputation and prediction tasks using real-world MFTS data. A flow chart of our method is shown in Figure~\ref{flowchart}

\begin{figure}[!ht]
\captionsetup{width=\linewidth}
 \centering
		\includegraphics[width=0.8\linewidth]{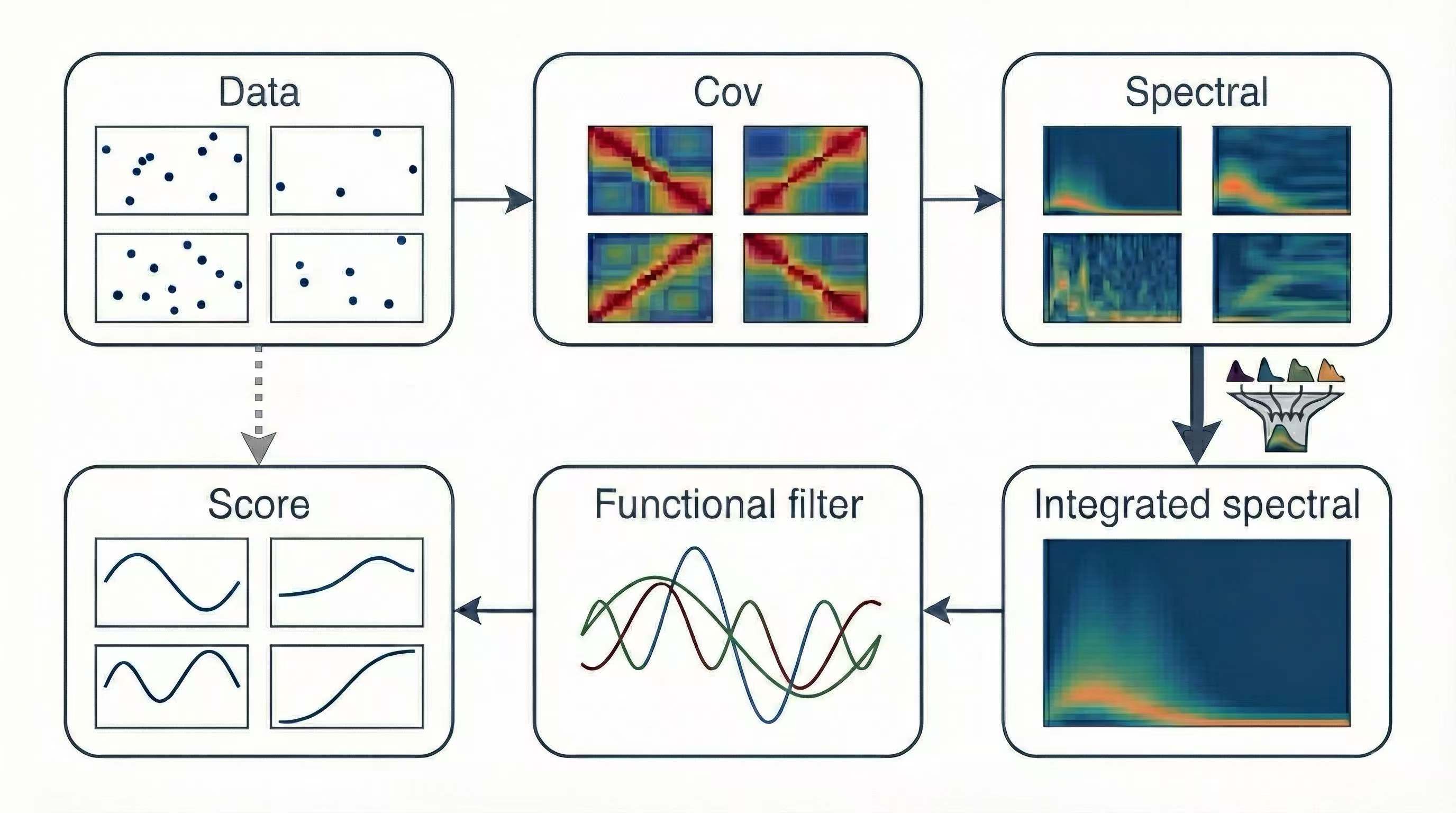}
	\caption{A flow chart of our method.}
 \label{flowchart}
\end{figure}

The remainder of this paper is organized as follows. 
Section~\ref{section_method} introduces the proposed Spectral MPCA framework for integrating multi-way dependencies in MFTS. 
Section~\ref{section_estimation} details the estimation procedures, including the marginalization of spectral density kernels and methods for data imputation and forecasting. 
Section~\ref{section_simulation} presents extensive simulation studies that evaluate the efficacy of Spectral MPCA and compare its performance with existing methods. 
Finally, Section~\ref{section_casestudy} demonstrates the practical utility of the proposed method through an application to PM$_{2.5}$ concentration data.

\section{Integrative Analysis of Multivariate Functional Time Series} \label{section_method}

\subsection{Notations}
Let $L^2\left(\mathcal{T}, \mathbb{C}^{p}\right)$ denote the Hilbert space of square-integrable functions from a compact domain $\mathcal{T}$ to the $p$-dimensional complex Euclidean space $\mathbb{C}^{p}$, equipped with the inner product $\langle \boldsymbol{f}, \boldsymbol{g}\rangle_{p}:=\int_{\mathcal{T}} \{\boldsymbol{f}(t)\}^{*} \boldsymbol{g}(t)\,\mathrm{d}t$ and the induced norm $\|\boldsymbol{f}\|_{p}:=\langle \boldsymbol{f}, \boldsymbol{f}\rangle_{p}^{1/2}$, $\forall\, \boldsymbol{f}, \boldsymbol{g} \in L^2\left(\mathcal{T}, \mathbb{C}^{p}\right)$, where $\{\cdot\}^{*}$ represents the conjugate transpose operation on a complex-valued vector or matrix. For convenience, we take $\mathcal{T} = [0,1]$ in what follows. When $p=1$, the notations $\langle \cdot, \cdot \rangle_{p}$ and $\|\cdot\|_{p}$ are simplified to $\langle \cdot, \cdot \rangle$ and $\|\cdot\|$, respectively. Furthermore, define $\mathcal{H}([0,1], \mathbb{R})$ as a collection of functions $\{ h_l(\cdot) \in L^2([0,1], \mathbb{R}) : l \in \mathbb{Z} \}$, for which there exists a function $H(\cdot \mid \omega) \in L^2([0,1], \mathbb{C})$ such that $\int_{-\pi}^{\pi} \|H(\cdot \mid \omega)\|^2 \,\mathrm{d}\omega < \infty$ and $\lim_{L \to \infty} \int_{-\pi}^{\pi} \|H(\cdot \mid \omega) - \frac{1}{2\pi} \sum_{|l| \leq L} h_l(\cdot)\exp(\mathrm{i} l \omega)\|^2 \,\mathrm{d}\omega = 0$, where $\mathrm{i}$ denotes the imaginary unit. Throughout, $|\cdot|$ and $\overline{(\cdot)}$ denote the modulus and the conjugate of a complex number or vector, $(\cdot)^{\top}$ represents the transpose of a vector or matrix, and $[c] := \{1, \ldots, c\}$ denotes the index set for any positive integer $c$.

Consider a $p$-dimensional MFTS $\boldsymbol{X}_j(t):=\{X_{1j}(t), \ldots, X_{pj}(t)\}^{\top}$ for $t \in [0,1]$, where $j \in \mathbb{Z}$ is a discrete-time index and each element satisfies $X_{ij}(\cdot) \in L^2\left([0,1], \mathbb{R}\right)$ for $i \in [p]$. For notational simplicity, we represent $\boldsymbol{X}_j(\cdot)$ and ${X}_{ij}(\cdot)$ as $\boldsymbol{X}_j$ and ${X}_{ij}$. Assume that $\{\boldsymbol{X}_j; j \in \mathbb{Z}\}$ is weakly stationary, i.e., $\operatorname{E}\{\boldsymbol{X}_j(t)\}$ and $\operatorname{cov}\{\boldsymbol{X}_{j+h}(t), \boldsymbol{X}_{j}(s)\}$ do not depend on $j$ for any $h \in \mathbb{Z}$ and $t,s \in [0,1]$. This assumption is widely used in MFTS studies \citep{hallin2023factor, chang2023autocovariance, tan2024graphical}. We define $\boldsymbol{\mu}(t) := \operatorname{E}\{\boldsymbol{X}_j(t)\}$ and $\boldsymbol{C}_{h}(t,s) := \operatorname{cov}\{\boldsymbol{X}_{j+h}(t), \boldsymbol{X}_{j}(s)\}$ for $t,s \in [0,1]$ and $j, h \in \mathbb{Z}$, which are the mean and autocovariance functions of $\{\boldsymbol{X}_j; j \in \mathbb{Z}\}$, respectively. Under the stationarity assumption, the spectral density functions $\{\boldsymbol{f}(\cdot, \cdot \mid \omega); \omega \in [-\pi, \pi]\}$ of MFTS are defined via the Fourier transform as
\begin{align}
\boldsymbol{f}(t, s \mid \omega)=\frac{1}{2 \pi} \sum_{h \in \mathbb{Z}} \boldsymbol{C}_h(t, s) \exp (\mathrm{i} h \omega), \quad t, s \in[0,1],\ \omega \in[-\pi, \pi].
\label{f to r}
\end{align}
To ensure convergence, we impose a summability condition as
$$\sup_{i_1, i_2 \in [p]} \sum_{h \in \mathbb{Z}}\left\{\int_0^1 \int_0^1\left|C_{i_1i_2h}(t, s)\right|^2 \mathrm{d} t \mathrm{d} s\right\}^{1 / 2}<\infty,$$
where $C_{i_1i_2h}(t, s)$ is the ($i_1,i_2$)-th element of $\boldsymbol{C}_h(t, s)$ for $h \in \mathbb{Z}$ and $t,s \in [0,1]$.
Conversely, the autocovariance functions $\boldsymbol{C}_h(\cdot, \cdot)$ can be represented via the inverse Fourier transform as
\begin{align}
\boldsymbol{C}_h(t, s)=\int_{-\pi}^\pi \boldsymbol{f}(t, s \mid \omega) \exp (-\mathrm{i} h \omega) \mathrm{d} \omega, \quad t, s \in[0,1],\ h \in\mathbb{Z}.
\label{r to f}
\end{align}
Similarly, we denote the ($i_1,i_2$)-th element of $\boldsymbol{f}(t, s \mid \omega)$ by ${f}_{i_1i_2}(t, s \mid \omega)$ for $\omega \in [-\pi,\pi]$ and $t,s \in [0,1]$.

\subsection{Integrating Multi-way Dependence in the Frequency Domain}
We define a zero-mean MFTS as $\boldsymbol{\varepsilon}_{j}(t) := \boldsymbol{X}_j(t) - \boldsymbol{\mu}(t)$ for $t \in [0,1]$ and $j \in \mathbb{Z}$, where $\boldsymbol{\varepsilon}_{j}(\cdot) = \{\varepsilon_{1j}(\cdot), \ldots, \varepsilon_{pj}(\cdot)\}^{\top}$. For brevity, we write $\boldsymbol{\varepsilon}_{j}(\cdot)$ and $\varepsilon_{ij}(\cdot)$ as $\boldsymbol{\varepsilon}_{j}$ and $\varepsilon_{ij}$, respectively. 
The series $\{\boldsymbol{\varepsilon}_{j}; j \in \mathbb{Z}\}$ exhibits complex multi-way dependence structures, including within-curve structures, temporal correlations (across $j$), and cross-subject dependencies (across $i$). These structures are incorporated in the spectral density \eqref{f to r}. Therefore, efficient integrated analysis of MFTS can be implemented via statistical procedures based on the spectral density.

Before introducing our methods, we provide a more detailed review of existing methods. 
Typically, conventional approaches for MFTS model each series $\varepsilon_{ij}$ separately and primarily target within-subject temporal dependence, constructing subject-specific models in either the time domain or the frequency domain.
In the time domain, one establishes the subject-specific models from the lag-$h$ covariances, for example, the aggregated lag-$h$ covariance $\sum_{h \in \mathbb{Z}} C_{iih}(t_1,t_2)$ \citep{gao2019high} or autocovariance integrals $\sum_{h \geq 1} \int_0^1 C_{iih}(t,t_1) C_{iih}(t,t_2)\ \mathrm{d}t$ \citep{bathia2010identifying, chang2024modelling}. In the frequency domain, existing methods mainly decompose spectral densities $f_{ii}(\cdot, \cdot \mid \omega)$ as 
\begin{align}
    f_{ii}(t, s \mid \omega) = \sum_{k=1}^{\infty} \eta_{ik}(\omega) \overline{\psi_{ik}(t \mid \omega)} \psi_{ik}(s \mid \omega), \quad t,s \in [0,1], \ \omega \in [-\pi, \pi],
    \label{uni_decom_f}
\end{align}
and then apply the inverse Fourier transform to the eigenfunctions $\psi_{ik}(\cdot \mid \omega)$ to obtain a set of basis functions, termed functional filters \citep{hormann2015dynamic}. 
This yields a representation for $\varepsilon_{ij}$ in the frequency domain via the dynamic KL expansion
\begin{align}
\varepsilon_{ij}(t) = \sum_{k=1}^{\infty}\sum_{l \in \mathbb{Z}} \phi_{ikl}(t)\xi_{i(j+l)k} 
\quad \text{with} \quad 
\xi_{ijk} = \sum_{l \in \mathbb{Z}} \langle \varepsilon_{i(j-l)}, \phi_{ikl} \rangle, \quad t \in [0,1], 
\label{uni_dfpca}
\end{align}
where $\{\phi_{ikl}(\cdot) ; l \in \mathbb{Z}\}$ are functional filters, and $\xi_{ijk}$ are the corresponding FPC scores. This representation \eqref{uni_dfpca} achieves an optimal approximation of $\varepsilon_{ij}$ in the mean squared error sense \citep{hormann2015dynamic} and unifies time- and frequency-domain approaches through optimal functional filtering tools developed in \citet{guo2024unified}. However, these approaches typically fail to exploit information across subjects, resulting in inefficient estimation and forecasting for MFTS, especially when the available within-subject information is limited.

To integrate information from all subjects, we construct a frequency-domain marginalization by defining the marginal spectral density function as
\begin{align}
    f_{\mathcal{S}}(t,s \mid \omega) := \frac{1}{p} \sum_{i \in [p]} f_{ii}(t,s \mid \omega), \quad t,s \in [0,1], \ \omega \in [-\pi, \pi].
    \label{margin_spec}
\end{align}
For each $\omega \in [-\pi, \pi]$, $f_{\mathcal{S}}(\cdot, \cdot \mid \omega)$ is a symmetric and positive-definite function; the proof is provided in Proposition \ref{marginal_sd_property} in the Appendix. Consequently, by Mercer's theorem, it admits the spectral decomposition
\begin{align}
    f_{\mathcal{S}}(t,s \mid \omega) = \sum_{k=1}^{\infty} \eta_k(\omega) \overline{\psi_k(t \mid \omega)} \psi_k(s \mid \omega), \quad t,s \in [0,1], \ \omega \in [-\pi,\pi],
    \label{margin_decom_f}
\end{align}
where $\eta_k(\omega)$ and $\psi_k(\cdot\mid \omega) \in L^2([0,1],\mathbb{C})$ denote the $k$th eigenvalue and eigenfunction, respectively. 

The proposed marginal spectral density function \eqref{margin_spec} jointly captures multi-way dependence structures within MFTS in the frequency domain, sharing the same spirit as MPCA in marginalizing lag-$0$ covariances \citep{chen2017modelling}. Specifically, the individual dependence for each subject $i$, including within-curve structures and temporal correlations, is captured by the corresponding spectral density $f_{ii}(\cdot,\cdot \mid \omega)$ as in \eqref{f to r}, and information across subjects is integrated through averaging as in \eqref{margin_spec}. Here, $f_{\mathcal{S}}(\cdot, \cdot \mid \omega)$ serves as the foundation for our framework, and its eigenfunctions $\psi_{k}(\cdot \mid \omega)$ in \eqref{margin_decom_f} borrow strength from all individuals, yielding shared functional filters that enable efficient integrative analysis of MFTS.

\subsection{Marginal Dynamic Karhunen--Lo\`{e}ve Expansion}\label{Section_MDKL}

Based on \eqref{margin_decom_f}, we introduce a novel representation for $\boldsymbol{\varepsilon}_{j}$, termed the marginal dynamic KL expansion:
\begin{align}
    \boldsymbol{\varepsilon}_{j}(t) = \sum_{k \geq 1} \sum_{l \in \mathbb{Z}} \phi_{kl}(t) \boldsymbol{\xi}_{\cdot(j+l)k}, \quad t \in [0,1],
    \label{marginal_dfpca}
\end{align}
with $\phi_{kl}(\cdot) = \frac{1}{2\pi} \int_{-\pi}^{\pi} \psi_k(\cdot \mid \omega) \exp(-\mathrm{i} l \omega)\,\mathrm{d} \omega$,
where $\psi_k(\cdot \mid \omega)$ are the eigenfunctions of the marginal spectral density $f_{\mathcal{S}}(\cdot, \cdot \mid \omega)$ as in \eqref{margin_decom_f}. 
Here, $\boldsymbol{\xi}_{\cdot jk} := (\xi_{1jk}, \ldots, \xi_{pjk})^{\top}$ with $\xi_{ijk} = \sum_{l \in \mathbb{Z}} \langle \varepsilon_{i(j-l)}, \phi_{kl} \rangle$ denotes the vector of the $k$th scores for $\boldsymbol{\varepsilon}_{j}$. For each $i$, the sequence $\boldsymbol{\xi}_{i\cdot k} := \{\xi_{ijk}; j \in \mathbb{Z}\}$ forms a zero-mean time series with spectral density
\begin{align}
    \tilde{\eta}_{ik}(\omega) := \int_0^1 \int_0^1 \psi_k(t \mid \omega) f_{ii}(t,s \mid  \omega) \overline{\psi_k(s \mid \omega)}\,\mathrm{d}t\,\mathrm{d}s, \quad \omega \in [-\pi,\pi], \ k \geq 1;
    \label{margin_score_spec}
\end{align}
see Proposition \ref{marginal_score_sd} in the Appendix for a detailed proof.

Notably, $\psi_k(\cdot \mid \omega)$ in \eqref{margin_decom_f} is unique only up to multiplication by a function on the complex unit circle $\mathcal{M}:=\big\{\nu:[-\pi,\pi] \rightarrow \mathbb{C} ; |\nu(\omega)| = 1, \nu(\omega) = \overline{\nu(-\omega)}\big\}$, i.e., for any $\nu_k(\cdot) \in \mathcal{M}$, $\psi_{k}(\cdot \mid \omega)\nu_{k}(\omega)$ is also an eigenfunction of $f_{\mathcal{S}}(\cdot,\cdot \mid \omega)$. This non-uniqueness issue may lead to redundant representations and degrade forecasting accuracy \citet{guo2024unified}. To avoid this issue, we extend the optimal functional filter criterion in \citep{guo2024unified} to the MFTS setting. Specifically, we select the optimal multiplier $\tilde{\nu}_k(\cdot)$ within the complex unit circle $\mathcal{M}$ as
\begin{align}
\mathop{\arg\max}_{\nu_k(\cdot)\in \mathcal{M}} \frac{1}{4 \pi^2}  \int_{-\pi}^{\pi} \int_{-\pi}^{\pi} \Psi_k(\omega_1, \omega_2) \overline{\nu_k(\omega_1)}\nu_k(\omega_2) \ \mathrm{d} \omega_1 \mathrm{d} \omega_2, \quad k \geq 1,
    \label{marginal_opt_ff_def2}
\end{align}
where the kernels $\{\Psi_k(\cdot, \cdot); k \geq 1\}$ are defined as $\Psi_k(\omega_1, \omega_2) = \int_{0}^1\overline{\psi_k(t \mid \omega_1)}\psi_k(t \mid \omega_2)\,\mathrm{d} t$ for $\omega_1, \omega_2 \in [-\pi, \pi]$, with $\psi_k(\cdot \mid \omega)$ being any valid $k$th eigenfunction of $f_{\mathcal{S}}(\cdot,\cdot\mid \omega)$. Accordingly, the functional filters $\phi_{kl}(\cdot)$ in \eqref{marginal_dfpca} are constructed as
\begin{align}
    \phi_{kl}(t)= \frac{1}{2\pi} \int_{-\pi}^{\pi}\psi_k(t \mid \omega)\tilde{\nu}_k(\omega) \exp (-\operatorname{i}l\omega)\,\mathrm{d} \omega, \quad t \in [0,1], \ l \in \mathbb{Z}.
    \label{marginal_opt_ff}
\end{align}

The proposed framework offers several advantages for modeling MFTS with multi-way dependence. First, compared with subject-specific models \citep{hormann2015dynamic,gao2019high, chang2024modelling, guo2024unified}, \eqref{marginal_dfpca} leverages cross-subject information to learn shared filters together with structured score processes, offering a parsimonious modeling framework for MFTS. Second, by modeling serial dependence in the frequency domain, it attains an $L^2$-optimal dimension reduction for MFTS, as established in Theorem \ref{opt_approximation} below, providing a suitable modeling framework for MFTS reconstruction and forecasting. We therefore call this framework Spectral Marginal Principal Component Analysis (Spectral MPCA).

\begin{theorem}
For any positive integer $K$, we define $\tilde{\boldsymbol{\varepsilon}}_{j}^{K} = (\tilde{\varepsilon}_{1j}^{K}, \ldots, \tilde{\varepsilon}_{pj}^{K})^{\top}$, where
\begin{align*}
    \tilde{\varepsilon}_{ij}^{K}(\cdot) := \sum_{k \in [K]} \sum_{l \in \mathbb{Z}} \tilde{\phi}_{kl}(\cdot) \tilde{\xi}_{i(j+l)k}, \quad \text{with} \ \ \tilde{\xi}_{ijk} = \sum_{l \in \mathbb{Z}} \langle \varepsilon_{i(j-l)}, \tilde{v}_{kl} \rangle, \ i \in [p].
\end{align*}
Here, $\{\tilde{\phi}_{kl}(\cdot); l \in \mathbb{Z}\}$ and $\{\tilde{v}_{kl}(\cdot); l \in \mathbb{Z}\}$ are sequences belonging to $\mathcal{H}([0,1], \mathbb{R})$. Then,
\begin{align*}
    \operatorname{E} \| \boldsymbol{\varepsilon}_{j} - \boldsymbol{\varepsilon}_{j}^{K} \|_{p}^{2} \leq \operatorname{E} \| \boldsymbol{\varepsilon}_{j} - \tilde{\boldsymbol{\varepsilon}}_{j}^{K} \|_{p}^{2},
\end{align*}
where $\boldsymbol{\varepsilon}^{K}_{j}(\cdot) = \sum_{k \in [K]} \sum_{l \in \mathbb{Z}} \phi_{kl}(\cdot) \boldsymbol{\xi}_{\cdot(j+l)k}$ denotes the truncation of the marginal dynamic KL expansion.
\label{opt_approximation}
\end{theorem}

\section{Estimation}\label{section_estimation}

We consider a latent MFTS with additive measurement errors, and the discrete observations satisfy
\begin{align}
    Y_{ijz} = X_{ij}(t_{ijz}) + \tau_{ijz} = \mu_i(t_{ijz}) + \varepsilon_{ij}(t_{ijz}) + \tau_{ijz}, \quad i \in [p], \ j \in [J], \ z \in [N_{ij}],
    \label{discrete_obs} 
\end{align}
where $\{t_{ijz} \in [0,1]; z \in [N_{ij}] \}$ denote the observation times for $X_{ij}$, $N_{ij}$ are measurement counts, $\tau_{ijz} \sim \mathcal{N}(0, \sigma_i^2)$ are i.i.d.\ Gaussian noise with variance $\sigma_i^2$, and $\mu_i(\cdot)$ denote the mean functions. Here, $t_{ijz}$ and $N_{ij}$ may vary across $i$ and $j$, reflecting various missing patterns in real-world data. The tasks for these irregular data usually include missing-data imputation to recover latent trajectories and prediction/forecasting to infer future functional observations from historical data. To this end, we estimate the mean functions $\mu_i(\cdot)$ and approximate the zero-mean process $\boldsymbol{\varepsilon}_j$ with the finite truncation of the marginal dynamic KL expansion \eqref{marginal_dfpca} as
\begin{align}
    \boldsymbol{\varepsilon}_j(\cdot) \approx \sum_{k \in [K]} \sum_{|l| \leq  L_k} \phi_{kl}(\cdot) \boldsymbol{\xi}_{\cdot(j+l)k},
    \label{truncate_marginal_dfpca}
\end{align}
where $K$ and $L_k$ are truncation parameters. In the following, we describe the estimation of the marginal functional filters $\phi_{kl}(\cdot)$ and the extraction of the corresponding scores $\boldsymbol{\xi}_{\cdot(j+l)k}$, which form the basis for downstream imputation and forecasting.

\subsection{Estimation of Marginal Functional Filters}\label{section_marginff_est}

For each $i \in [p]$, we adopt local linear smoothers to estimate the mean functions $\mu_i(\cdot)$ from the discrete observations $Y_{ijz}$, and then estimate the autocovariance functions $C_{iih}(\cdot, \cdot)$ based on the products 
\(
\{ Y_{i(j+h)z_{1}} - {\mu}_{i}(t_{i(j+h)z_{1}}) \} \{ Y_{ijz_{2}} - {\mu}_{i}(t_{ijz_{2}}) \}
\) for $j$ such that $j \in [J]$ and $j+h \in [J]$, and $z_1 \in [N_{i(j+h)}]$ and $z_2 \in [N_{ij}]$, in a similar manner; refer to \citet{yao2005functional, rubin2020sparsely, guo2024unified} for detailed procedures. 
Local linear smoothing is widely adopted in functional data analysis, and many established statistical software packages provide implementations \citep{jacoby2000loess, cleveland2017local}. 
In our implementation, we use the \texttt{loess} function in R and denote the estimators of ${\mu}_{i}(\cdot)$ and ${C}_{iih}(\cdot, \cdot)$ by $\hat{\mu}_{i}(\cdot)$ and $\hat{C}_{iih}(\cdot, \cdot)$, respectively.

To estimate the marginal functional filters $\phi_{kl}(\cdot)$ in \eqref{marginal_opt_ff}, we first estimate the individual spectral densities using the Bartlett lag-window estimators \citep{brillinger2001time} for each $i \in [p]$ as
\begin{align}
    \hat{f}_{ii}(t,s \mid \omega) = \frac{1}{2\pi} \sum_{|h| \leq h_{\text{max}}} \bigg(1 - \frac{|h|}{h_{\text{max}}} \bigg) \hat{C}_{iih}(t, s) \exp(\mathrm{i} h\omega), \quad  t,s \in [0,1], \ \omega \in [-\pi,\pi],
    \label{individual_spec_est}
\end{align}
where $h_{\text{max}}$ denotes the truncation parameter. Next, we estimate the marginal spectral density $f_{\mathcal{S}}(\cdot,\cdot \mid \omega)$ in \eqref{margin_spec} via the plug-in estimator
\begin{align}
    \hat{f}_{\mathcal{S}}(t,s \mid \omega) = \frac{1}{p} \sum_{i \in [p]} \hat{f}_{ii}(t,s \mid \omega), \quad t,s \in [0,1], \ \omega \in [-\pi,\pi].
    \label{margin_spec_est}
\end{align}
Then, we perform a spectral decomposition of $\hat{f}_{\mathcal{S}}(\cdot,\cdot \mid \omega)$ to obtain eigenfunctions $\hat{\psi}_k(\cdot \mid \omega)$ for $k \in [K]$. In practice, for each $\omega \in [-\pi,\pi]$, we discretize the functional domain $[0,1]$ on a dense grid $\{t_m\}_{m=1}^{M_t}$ and form the $M_t \times M_t$ matrix $\big\{\hat{f}_{\mathcal{S}}(t_m,t_{m'} \mid \omega)\big\}_{m,m'=1}^{M_t}$. We then compute its eigendecomposition to obtain numerical approximations of the eigenfunctions. Repeating this step over a grid of $M_\omega$ frequency points yields an overall computational complexity of $\mathcal{O}(M_\omega M_t^3)$, where $M_t$ and $M_\omega$ are the numbers of grid points used to discretize $[0,1]$ and $[-\pi,\pi]$, respectively.

To solve the constrained optimization problem in \eqref{marginal_opt_ff_def2}, we replace $\Psi_k(\cdot, \cdot)$ with its estimator $\hat{\Psi}_k(\cdot, \cdot)$, where $\hat{\Psi}_k(\omega_1, \omega_2) = \int_0^1 \overline{\hat{\psi}_k(t \mid \omega_1)} \hat{\psi}_k(t \mid \omega_2)\,\mathrm{d} t$ for $\omega_1, \omega_2 \in [-\pi, \pi]$, and implement the projected gradient algorithm in \citet{guo2024unified} to obtain multiplicative factors $\hat{\nu}_k(\cdot)$. Finally, we construct the marginal functional filters from the estimated multiplicative factor $\hat{\nu}_k(\cdot)$ and the eigenfunctions $\hat{\psi}_k(\cdot \mid \omega)$ as
\begin{align}
    \hat{\phi}_{kl}(t)= \frac{1}{2\pi} \int_{-\pi}^{\pi} \hat{\psi}_k(t \mid \omega)\hat{\nu}_k(\omega) \exp (-\mathrm{i}l\omega)\,\mathrm{d} \omega, \quad t \in [0,1], \ l \in \mathbb{Z}.
    \label{marginal_opt_ff_est}
\end{align}

\subsection{Score Extraction and Downstream Tasks}\label{section_score_est}
We adopt a Bayesian approach to extract the latent scores and facilitate downstream tasks in this section. Let $\boldsymbol{Y} := \{Y_{ijz}; i \in [p], j \in [J], z \in [N_{ij}]\}$ represent the discrete observations. For $k \in [K]$, denote $\boldsymbol{\xi}_{\cdot \cdot k} := \{\boldsymbol{\xi}_{\cdot (1-L_k)k}, \ldots, \boldsymbol{\xi}_{\cdot (J+L_k)k} \}$ as the scores in the truncated model \eqref{truncate_marginal_dfpca}. By Bayes' theorem, the posterior distribution of $\{\boldsymbol{\xi}_{\cdot \cdot k}; k \in [K]\}$ can be represented as
\begin{equation}
\pi\left(\boldsymbol{\xi}_{\cdot \cdot 1}, \ldots, \boldsymbol{\xi}_{\cdot \cdot K} \mid \boldsymbol{Y}\right)
\propto \mathcal{L}(\boldsymbol{\xi}_{\cdot \cdot 1}, \ldots, \boldsymbol{\xi}_{\cdot \cdot K} \mid \boldsymbol{Y}) \cdot \prod_{k \in [K]} \pi(\boldsymbol{\xi}_{\cdot \cdot k}),
\label{posterior}
\end{equation}
where $\pi(\cdot)$ denotes the prior distribution. Under the Gaussian assumption, we have
\begin{align*}
&\mathcal{L}(\boldsymbol{\xi}_{\cdot \cdot 1}, \ldots, \boldsymbol{\xi}_{\cdot \cdot K} \mid \boldsymbol{Y}) \\
&\propto 
\exp \left[
- \sum_{i \in [p]} \sum_{j \in [J]} \sum_{z \in [N_{ij}]}
\frac{
\left\{Y_{ijz} - \mu_i(t_{ijz}) - \sum_{k \in [K]} \sum_{|l| \leq L_k} \phi_{kl}(t_{ijz}) \, \xi_{i(j+l)k}\right\}^2
}{2\sigma_i^2}
\right],
\end{align*}
where $\mu_i(\cdot)$ and $\phi_{kl}(\cdot)$ are replaced with their estimates $\hat{\mu}_i(\cdot)$ and $\hat{\phi}_{kl}(\cdot)$, respectively, and $\sigma_i^2$ can be estimated from the observed data $\boldsymbol{Y}$; see the approach in \citet{yao2005functional} for more details.  

Notice that the scores $\boldsymbol{\xi}_{i \cdot k} =  \{{\xi}_{i jk}; j \in \mathbb{Z}\}$ in \eqref{marginal_dfpca} form a weakly stationary time series with spectral density $\tilde{\eta}_{ik}(\cdot)$ defined in \eqref{margin_score_spec}; we employ the Whittle likelihood \citep{whittle1961gaussian, subba2021reconciling} to construct prior distributions $\pi(\boldsymbol{\xi}_{\cdot \cdot k})$ in the frequency domain. Specifically, we define the discrete Fourier transform of $\boldsymbol{\xi}_{\cdot \cdot k}$ as $\tilde{\boldsymbol{\xi}}_{\cdot k}(\omega_j) := \frac{1}{\sqrt{2\pi (J+2L_k)}} \sum_{j \in [J+2L_k]} \exp( \mathrm{i} j\omega_j) \boldsymbol{\xi}_{\cdot (j-L_k) k}$ for $ \omega_j \in \mathcal{S}_J$,
where $\mathcal{S}_J := \{\omega_j = 2\pi j/J : j \in [J]\}$. Following the asymptotic properties of $\tilde{\boldsymbol{\xi}}_{\cdot k}(\cdot)$ \citep{whittle1961gaussian, subba2021reconciling}, the log-transformed prior distribution takes the form:
\begin{equation*}
\log \pi(\boldsymbol{\xi}_{\cdot \cdot k}) = -\frac{1}{2} \sum_{j \in [J]} \left[
\left\{ \tilde{\boldsymbol{\xi}}_{\cdot k}(\omega_j) \right\}^* \boldsymbol{\Phi}_k(\omega_j) \tilde{\boldsymbol{\xi}}_{\cdot k}(\omega_j)
- \log\det\left\{ \boldsymbol{\Phi}_k(\omega_j) \right\}
\right], \quad \omega_j \in \mathcal{S}_J.
\end{equation*}
Here, $\boldsymbol{\Phi}_k(\omega_j)$ is a diagonal matrix with elements $\{\tilde{\eta}^{-1}_{1k}(\omega_j), \ldots, \tilde{\eta}^{-1}_{pk}(\omega_j)\}$, which can be estimated directly from $\hat{\psi}_k(\cdot \mid \omega)$ and $\hat{f}_{ii}(\cdot,\cdot \mid  \omega)$ via \eqref{margin_score_spec}.

Given truncation orders $K$ and $\{L_k\}_{k \in [K]}$, the log-likelihood $\log \mathcal{L}(\boldsymbol{\xi}_{\cdot \cdot 1}, \ldots, \boldsymbol{\xi}_{\cdot \cdot K} \mid \boldsymbol{Y})$ in \eqref{posterior} can be written in matrix form as
\begin{align}
    - \frac{1}{2}\,\big\|\bm{W}^{1/2}\big(\bm{y}-\bm{A}\bm{\xi}\big)\big\|_{\text{E}}^2 + C,
    \label{mat_log_likelihood}
\end{align}
where $\| \cdot \|_{\text{E}}$ is the Euclidean norm for vectors, and $C$ is a constant that is independent of $(\boldsymbol{\xi}_{\cdot \cdot 1}, \ldots, \boldsymbol{\xi}_{\cdot \cdot K})$. Here, $\bm{y}$ denotes the vectorized demeaned observations $\{\tilde{Y}_{ijz} = Y_{ijz} - \mu_i(t_{ijz}); i \in [p], j \in [J], z \in [N_{ij}]\}$, and $\bm{W}$ denotes the diagonal matrix with weights $\sigma_i^{-2}$ aligned with $\bm{y}$. Moreover, $\bm{\xi}$ denotes the vectorized scores $\{\boldsymbol{\xi}_{\cdot \cdot k}; k \in [K]\}$, and $\bm{A}$ is the design matrix consisting of functional filters $\phi_{kl}(\cdot)$ such that each element of $\bm{A}\bm{\xi}$ takes the form $\sum_{k \in [K]} \sum_{|l| \leq L_k} \phi_{kl}(t_{ijz}) \, \xi_{i(j+l)k}$. 
Similarly, the joint log-prior can be written in matrix form as
\begin{align}
    \sum_{k \in [K]} \log \pi(\boldsymbol{\xi}_{\cdot \cdot k}) = - \frac{1}{2}\,\bm{\xi}^{*}\bm{Q}\,\bm{\xi} + C,
    \label{mat_log_prior}
\end{align}
where 
\begin{align*}
    \bm{Q}=\mathrm{diag}(\bm{Q}_1,\ldots,\bm{Q}_K), \quad \text{with} \ \bm{Q}_k=(\bm{F}_k\!\otimes\! \bm{I}_p)^{*}\!\left(\bigoplus_{\omega_j\in\mathcal{S}_J}\boldsymbol{\Phi}_k(\omega_j)\right)\!(\bm{F}_k\!\otimes\! \bm{I}_p).
\end{align*}
Here, $\bm{I}_p$ denotes the $p \times p$ identity matrix, and $\bm{F}_k$ denotes the discrete Fourier transform matrix with entries $[\bm{F}_k]_{j,r} = \frac{1}{\sqrt{2\pi (J+2L_k)}} \exp(\mathrm{i} r \omega_j)$, where $[\cdot]_{j,r}$ denotes the $(j,r)$-th entry of a matrix, $\otimes$ denotes the Kronecker product, and $\bigoplus_{\omega_j\in\mathcal{S}_J}\boldsymbol{\Phi}_k(\omega_j)$ denotes the block-diagonal (direct-sum) matrix with diagonal blocks $\{\boldsymbol{\Phi}_k(\omega_j): \omega_j \in \mathcal{S}_J\}$. Given these matrix-vector expressions, we apply a gradient ascent algorithm to obtain the Maximum A Posteriori (MAP) estimator for the scores $\{\boldsymbol{\xi}_{\cdot \cdot k}; k \in [K]\}$ from \eqref{posterior}, denoting it by $(\hat{\boldsymbol{\xi}}_{\cdot \cdot 1}, \ldots, \hat{\boldsymbol{\xi}}_{\cdot \cdot K})$.

A complete algorithm for the above procedure is summarized in Algorithm~\ref{alg:mfts_estimation}, and we provide a detailed explanation and an illustrative example for the optimization \eqref{mat_MAP} in \partref{mat_posterior} of the Appendix. 
Here, the selection of truncation numbers $K$, $\{L_k\}_{k=1}^K$, and the maximum covariance lag $h_{\max}$ is discussed in \partref{parameter_selection} of the Appendix, and additional technical details for the score extraction procedure can be found in \partref{MAP_details} of the Appendix.  
In each iterative step of this optimization, the computational complexity is $\mathcal{O}(K p J^2)$.

\begin{algorithm}[!ht]
\caption{Estimation of marginal functional filters and scores.}
\label{alg:mfts_estimation}
\begin{adjustwidth}{-1.5em}{0em}
\begin{algorithmic}[1]
\Statex \textbf{Input:} Discrete observations $\boldsymbol{Y}=\{(t_{ijz}, Y_{ijz})\}$; truncation orders $K$ and $\{L_k\}_{k \in [K]}$; 
\Statex maximum covariance lag $h_{\max}$.
\Statex \textbf{Output:} Marginal functional filters $\{\hat{\phi}_{kl}(t)\!:\!|l|\le L_k,\,k\in[K]\}$ and scores $\{\hat{\boldsymbol{\xi}}_{\cdot\cdot k}\}_{k \in [K]}$.

\State \textbf{Mean removal}: For each component $i$, smooth $\{(t_{ijz},Y_{ijz})\}$ to obtain $\hat{\mu}_i(\cdot)$ and set $\tilde{Y}_{ijz}\gets Y_{ijz}-\hat{\mu}_i(t_{ijz})$.

\State \textbf{Autocovariances}: For lags $|h|\le h_{\max}$, estimate the covariance surface $\hat{C}_{iih}(t,s)$ from
\Statex products of mean-removed observations at time pairs $(t,s)$.

\State \textbf{Per-component spectral densities}: For $\omega\in[-\pi,\pi]$, 
\begin{align*}
    \hat{f}_{ii}(t,s\mid\omega)=\frac{1}{2\pi}\sum_{|h|\le h_{\max}}\!\Big(1-\frac{|h|}{h_{\max}}\Big)\hat{C}_{iih}(t,s)\, \exp{(\mathrm{i}h\omega)}.
\end{align*}

\State \textbf{Marginal spectral density}: $\displaystyle \hat{f}_{\mathcal S}(t,s\mid\omega)=\frac{1}{p}\sum_{i \in [p]}\hat{f}_{ii}(t,s\mid\omega).$

\State \textbf{Frequency-wise eigensystem}: For each $\omega$, eigendecompose $\hat{f}_{\mathcal S}(\cdot,\cdot\mid\omega)$ to obtain $\{\hat{\psi}_k(\cdot\mid\omega)\}_{k \in [K]}$.

\State \textbf{Filter construction}: For each $k$, solve for an optimal scalar factor $\hat{\nu}_k(\omega)$ from \eqref{marginal_opt_ff_def2} by 
\Statex replacing ${\Psi}_k(\omega_1, \omega_2)$ with $\int_0^1 \overline{\hat{\psi}_k(t \mid \omega_1)} \hat{\psi}_k(t \mid \omega_2)\,\mathrm{d} t$. Then set
\begin{align*}
    \displaystyle \hat{\phi}_{kl}(t)=\frac{1}{2\pi}\int_{-\pi}^{\pi}\hat{\psi}_k(t\mid\omega)\,\hat{\nu}_k(\omega)\,\exp{(-\mathrm{i}l\omega)}\,\mathrm{d}\omega,\quad |l|\le L_k.
\end{align*}

\State \textbf{Variance}: Estimate $\hat{\sigma}_i^2$ from residuals $\tilde{Y}_{ijz}$ using local variance smoothing \citep{yao2005functional}.

\State \textbf{Prior precision matrix construction}: Construct the prior precision matrix $\bm{Q}$ by
\Statex replacing $\{\boldsymbol{\Phi}_k(\cdot); k \in [K]\}$ with their estimates.

\State \textbf{Score extraction}: The scores are obtained by solving
\begin{align}
\hat{\bm\xi}\;=\;\arg\min_{\bm\xi}\ \frac{1}{2}\,\big\|\bm{W}^{1/2}\big(\bm{y}-\bm{A\xi}\big)\big\|_{\text{E}}^2\;+\;\frac{1}{2}\,\bm{\xi}^{*}\bm Q\,\bm \xi.
    \label{mat_MAP}
\end{align}
\end{algorithmic}
\end{adjustwidth}
\end{algorithm}

The data imputation and forecasting tasks can be conducted using the proposed estimators. For imputation, we reconstruct the entire curves $\boldsymbol{X}_j$ from incomplete observations $\boldsymbol{Y}$ as
\begin{align*}
    \hat{\boldsymbol{X}}_j(t) = \hat{\boldsymbol{\mu}}(t) + \sum_{k \in [K]} \sum_{|l| \leq L_k} \hat{\phi}_{kl}(t) \hat{\boldsymbol{\xi}}_{\cdot(j+l)k}, \quad t \in [0,1], \ j \in [J],
\end{align*}
where $\hat{\boldsymbol{\mu}}(\cdot) = (\hat{\mu}_1(\cdot), \ldots,\hat{\mu}_p(\cdot))^{\top}$ is the vector of estimated mean functions. For forecasting, we fit multivariate time series models, such as vector autoregression (VAR) or vector ARMA (VARMA) \citep{brillinger2001time, brockwell2009time}, to the estimated score processes $\{\hat{\boldsymbol{\xi}}_{\cdot \cdot k}\}_{k \in [K]}$. Given the predicted scores $\{\hat{\boldsymbol{\xi}}_{\cdot (J+L_k+1) k}, \ldots, \hat{\boldsymbol{\xi}}_{\cdot (J+L_k+P) k}\}$, the $m$-step-ahead forecast curves are
\begin{align*}
        \hat{\boldsymbol{X}}_{J+m}(t) = \hat{\boldsymbol{\mu}}(t) + \sum_{k \in [K]} \sum_{|l| \leq L_k} \hat{\phi}_{kl}(t) \hat{\boldsymbol{\xi}}_{\cdot(J+m+l)k}, \quad t \in [0,1], \ m \in [P],
\end{align*}
where $P$ denotes the forecast horizon.

\section{Simulation}\label{section_simulation}

In this section, we compare our proposed method with common MFTS analysis approaches for data imputation and forecasting. We generate $p$-dimensional zero-mean MFTS through the dynamic KL expansions
\begin{align}
    \varepsilon_{ij}(t) = \sum_{k \in [K]} \sum_{|l| \leq L_{k}} w_l \phi_{ikl}(t) {\xi}_{i (j+l)k}, \quad j \in [J], \ t \in [0,1],
    \label{sim_gen_MFTS}
\end{align}
where $K$ and $L_{k}$ are positive integers, and $J$ denotes the sample length. The weights $w_{l}$ are defined as $w_{l}=\sqrt{ w_{l}^{\prime} / \sum_{|l| \leq L_{k}} w_{l}^{\prime}}$ with $w_{l}^{\prime}=\exp (-|l|/2)$ for each $k$, ensuring $\sum_{|l| \leq L_{k}} w_l^2=1$ and assigning decreasing weights to $\phi_{ikl}(\cdot)$ as $|l|$ increases. For each $k$, define the $p$-dimensional score vector $\boldsymbol{\xi}_{\cdot jk}:=(\xi_{1jk},\ldots,\xi_{pjk})^{\top}$ for $j \in \mathbb{Z}$. The basis functions $\phi_{ikl}(\cdot)$ are constructed using Fourier bases $\tilde\phi_{kl}(\cdot)$ multiplied by a subject-specific fluctuation $1 + \sin (it/p)$ for each $i$. Serial dependencies are incorporated via a VAR(1) model for the scores in each component:
$\boldsymbol{\xi}_{\cdot (j+1)k} = \rho_k \boldsymbol{\xi}_{\cdot jk} + \boldsymbol{b}_{\cdot jk}$, where $\boldsymbol{b}_{\cdot jk}$ are independent across both $j$ and $k$. To introduce cross-individual dependence structures, we generate the Gaussian random variables via $\boldsymbol{b}_{\cdot jk} \sim \mathcal{N}(\boldsymbol{0}, (\boldsymbol{\Theta}_k^b)^{-1})$ \citep{zapata2022partial, tan2024graphical}, where the details of the inverse covariance matrix $\boldsymbol{\Theta}_k^b$ are provided in \partref{score_gen} of the Appendix. In our simulation setup, we set $p = 5$, $K = L_k = 1$, and $\rho_k = 0.5$.

We generate discrete MFTS observations $Y_{ijz}$ via
\begin{align}
    Y_{ijz} =  \varepsilon_{ij}(t_{ijz}) + \tau_{ijz}, \quad i \in [p], \ j \in [J], \ z \in [N_{ij}].
    \label{sim_gen_discrete}
\end{align}
Specifically, we construct an evenly spaced grid on $[0,1]$ with 31 candidate observation points. The observation count $N_{ij}$ is sampled from discrete uniform distributions over $\{4,5\}$, $\{5,\ldots,10\}$, and $\{10,\ldots,15\}$, corresponding to three different missingness levels. Given $N_{ij}$, the observation times $t_{ijz}$ are then randomly selected without replacement from the 31 candidate points. To introduce measurement error, we generate $\tau_{ijz}$ from $\mathcal{N}(0,\operatorname{E}\|\varepsilon_{i1}\|^2/10)$ independently across $i$, $j$, and $z$. In the following simulations, we consider sample lengths $J = 30, 60, 90$.

The above setting is referred to as Case~1. In addition, we conduct two further simulation scenarios considering heavy-tailed measurement errors and non-linear score dynamics. Specifically, for Case~2, we replace the Gaussian errors $\tau_{ijz}$ in \eqref{sim_gen_discrete} with rescaled $t$-distributed errors $c_i \cdot \tilde{\tau}_{ijz}$, where $\tilde{\tau}_{ijz}$ are sampled independently from the $t_{\nu}$ distribution with $\nu = 5$ and
$c_i = \sqrt{\frac{(\nu - 2)}{\nu} \cdot \frac{\operatorname{E}\|\varepsilon_{i1}\|^2}{10}}$, 
so that the resulting variances are comparable to those under the Gaussian errors. For Case~3, we generate the scores via the non-linear mechanism
$\boldsymbol{\xi}_{\cdot (j+1)k} = \rho_k \boldsymbol{\xi}_{\cdot jk} + \sin(\boldsymbol{\xi}_{\cdot jk}) + \boldsymbol{b}_{\cdot jk}$.
The remaining data generation settings are the same as in Case~1.

We compare four competing methods with our proposed method for data imputation and forecasting. The first two approaches, dynamic FPCA (DFPCA) \citep{hormann2015dynamic} and Principal Analysis via Dependency-Adaptivity (PADA) \citep{guo2024unified}, adopt frequency-domain representations for each subject in MFTS, achieving theoretical optimality under mild conditions. Both methods are built upon individual spectral density functions $f_{ii}(\cdot,\cdot\mid\omega)$, and PADA further selects optimal functional filters according to specified criteria. 
The third and fourth methods both operate in the time domain, namely Long-Run FPCA (LRFPCA) \citep{gao2019high} and Variational Multivariate FPCA (VMFPCA) \citep{chang2024modelling}. These approaches construct subject-specific dimension-reduction representations and then integrate information across subjects when modeling the resulting scores for MFTS prediction.
 
Notably, both PADA and Spectral MPCA accommodate discrete observations directly. In comparison, the other methods (DFPCA, LRFPCA, and VMFPCA) typically require pre-smoothing to recover complete curves. To accomplish MFTS forecasting, we fit AR models to the scores produced by DFPCA and PADA as in \citet{guo2024unified}. In contrast, scores from the remaining approaches are modeled using VAR models to capture cross-series dependencies. Lag orders for all time series models (AR and VAR) are selected by the Akaike information criterion (AIC) \citep{brockwell2009time}. Finally, forecasts of the MFTS are reconstructed from the predicted scores using the corresponding dimension-reduction representations for each method.

We assess data imputation performance using the normalized mean squared error (NMSE)
\begin{align*}
    \text{NMSE} 
        = \frac{\sum_{i \in [p]}\sum_{j \in [J]} \left\| \varepsilon_{ij} - \hat{\varepsilon}_{ij} \right\|^{2}}
               {\sum_{i \in [p]}\sum_{j \in [J]} \left\| \varepsilon_{ij} \right\|^{2}},
\end{align*}
where $\hat{\varepsilon}_{ij}$ is obtained from the competing methods, and the number of retained components $K$ is fixed at its true value in~\eqref{sim_gen_MFTS}. 
Forecasting performance is evaluated with the one-step-ahead criterion \citep{aue2015prediction, tang2022clustering}. Specifically, we generate $\{\boldsymbol{\varepsilon}_j; j \in [J+P]\}$ and discrete observations $\{Y_{ijz}; j \in [J + P], z \in [N_{ij}]\}$ from \eqref{sim_gen_MFTS} and~\eqref{sim_gen_discrete}, and then define the resulting normalized mean squared prediction error (NMSPE) as
\begin{align*}
    \text{NMSPE} 
        = \frac{\sum_{i \in [p]}\sum_{m \in [P]}
                \left\| \varepsilon_{i(J+m)} 
                          - \hat{\varepsilon}_{i\{(J+m)\mid 1:(J+m-1)\}} \right\|^{2}}
               {\sum_{i \in [p]}\sum_{m \in [P]} \left\| \varepsilon_{i(J+m)} \right\|^{2}},
\end{align*}
where $\hat{\varepsilon}_{i\{(J+m)\mid 1:(J+m-1)\}}$ denotes the one-step-ahead forecast based on observations $\{Y_{ijz}; j \in [J+m-1], \; z \in [N_{ij}]\}$. In this simulation, we set $P=5$.

We conduct 100 Monte Carlo simulations for each experimental setting. Table~\ref{tab:combined_cases_shared_rows} reports NMSEs for data imputation in the three cases. We make three observations. First, methods that do not rely on the pre-smoothing strategy, namely PADA and Spectral MPCA, outperform competing approaches that use pre-smoothing, reflecting the advantage of operating directly on discrete observations rather than on reconstructed trajectories. Second, Spectral MPCA yields smaller NMSEs than the time-domain approaches (LRFPCA and VMFPCA), even in the dense case ($N_{ij} \in \{10,\ldots,15\}$). This is because the representations used by time-domain approaches are generally not $L^2$-optimal for MFTS, consistent with the theoretical optimality of frequency-domain representations (Theorem~\ref{opt_approximation}). Third, Spectral MPCA imputes more accurately than methods relying on subject-specific representations (DFPCA and PADA), because its shared functional filters borrow strength across subjects and improve statistical efficiency. Taken together, these results demonstrate Spectral MPCA’s superior performance in data imputation: it unifies frequency-domain optimality and cross-subject information integration, fully exploiting multi-level dependence information while avoiding information loss induced by pre-smoothing.

\begin{sidewaystable}
\centering
\caption{The NMSEs of data imputation under three cases.}
\label{tab:combined_cases_shared_rows}
\resizebox{0.98\textwidth}{!}{\begin{tabular}{ccccc|ccc|ccc}
\toprule
\multicolumn{2}{c}{} &
\multicolumn{3}{c|}{\textbf{Case 1}} &
\multicolumn{3}{c|}{\textbf{Case 2}} &
\multicolumn{3}{c}{\textbf{Case 3}} \\
\cmidrule(lr){3-5}\cmidrule(lr){6-8}\cmidrule(lr){9-11}
\multicolumn{2}{c}{} &
\multicolumn{1}{c}{$N_{ij}\in\{4,5\}$} &
\multicolumn{1}{c}{$N_{ij}\in\{5,\ldots,10\}$} &
\multicolumn{1}{c|}{$N_{ij}\in\{10,\ldots,15\}$} &
\multicolumn{1}{c}{$N_{ij}\in\{4,5\}$} &
\multicolumn{1}{c}{$N_{ij}\in\{5,\ldots,10\}$} &
\multicolumn{1}{c|}{$N_{ij}\in\{10,\ldots,15\}$} &
\multicolumn{1}{c}{$N_{ij}\in\{4,5\}$} &
\multicolumn{1}{c}{$N_{ij}\in\{5,\ldots,10\}$} &
\multicolumn{1}{c}{$N_{ij}\in\{10,\ldots,15\}$} \\
\midrule

\multirow{6}{*}{\textbf{$\boldsymbol{J=30}$}}
& DFPCA         & 0.926 & 0.593 & 0.323 & 0.927 & 0.577 & 0.334 & 0.913 & 0.587 & 0.315 \\
& PADA          & 0.308 & 0.190 & 0.107 & 0.306 & 0.188 & 0.107 & 0.298 & 0.187 & 0.105 \\
& LRFPCA        & 0.905 & 0.794 & 0.728 & 0.906 & 0.791 & 0.741 & 0.894 & 0.785 & 0.724 \\
& VMFPCA        & 0.983 & 0.737 & 0.638 & 0.981 & 0.736 & 0.657 & 0.971 & 0.728 & 0.627 \\
& Spectral MPCA & 0.172 & 0.117 & 0.069 & 0.166 & 0.116 & 0.069 & 0.172 & 0.116 & 0.069 \\
\midrule

\multirow{6}{*}{\textbf{$\boldsymbol{J=60}$}}
& DFPCA         & 0.867 & 0.528 & 0.229 & 0.860 & 0.527 & 0.223 & 0.868 & 0.522 & 0.228 \\
& PADA          & 0.205 & 0.124 & 0.074 & 0.205 & 0.125 & 0.073 & 0.203 & 0.124 & 0.073 \\
& LRFPCA        & 0.907 & 0.820 & 0.762 & 0.903 & 0.818 & 0.763 & 0.908 & 0.817 & 0.762 \\
& VMFPCA        & 0.900 & 0.690 & 0.580 & 0.891 & 0.685 & 0.578 & 0.893 & 0.685 & 0.574 \\
& Spectral MPCA & 0.117 & 0.082 & 0.055 & 0.118 & 0.083 & 0.055 & 0.117 & 0.083 & 0.056 \\
\midrule

\multirow{6}{*}{\textbf{$\boldsymbol{J=90}$}}
& DFPCA         & 0.859 & 0.488 & 0.203 & 0.851 & 0.487 & 0.218 & 0.853 & 0.484 & 0.203 \\
& PADA          & 0.167 & 0.097 & 0.063 & 0.168 & 0.096 & 0.063 & 0.165 & 0.096 & 0.064 \\
& LRFPCA        & 0.908 & 0.824 & 0.782 & 0.910 & 0.824 & 0.782 & 0.910 & 0.824 & 0.782 \\
& VMFPCA        & 0.866 & 0.674 & 0.568 & 0.868 & 0.674 & 0.575 & 0.852 & 0.668 & 0.558 \\
& Spectral MPCA & 0.105 & 0.066 & 0.052 & 0.105 & 0.066 & 0.052 & 0.105 & 0.067 & 0.053 \\
\bottomrule
\end{tabular}
}
\end{sidewaystable}

Table~\ref{tab:NMSPE_three_cases} reports NMSPEs for evaluating forecasting performance. The proposed Spectral MPCA consistently achieves lower prediction errors across all simulation scenarios, demonstrating its practical applicability and reliability under heavy-tailed noise and both linear and non-linear dynamics in MFTS data.

\begin{sidewaystable}
\centering
\caption{The NMSPEs of MFTS forecasting under three cases.}
\label{tab:NMSPE_three_cases}
\resizebox{0.98\textwidth}{!}{
\begin{tabular}{ccccc|ccc|ccc}
\toprule
\multicolumn{2}{c}{} &
\multicolumn{3}{c|}{\textbf{Case 1}} &
\multicolumn{3}{c|}{\textbf{Case 2}} &
\multicolumn{3}{c}{\textbf{Case 3}} \\
\cmidrule(lr){3-5}\cmidrule(lr){6-8}\cmidrule(lr){9-11}
\multicolumn{2}{c}{} &
\multicolumn{1}{c}{$N_{ij}\in\{4,5\}$} &
\multicolumn{1}{c}{$N_{ij}\in\{5,\ldots,10\}$} &
\multicolumn{1}{c|}{$N_{ij}\in\{10,\ldots,15\}$} &
\multicolumn{1}{c}{$N_{ij}\in\{4,5\}$} &
\multicolumn{1}{c}{$N_{ij}\in\{5,\ldots,10\}$} &
\multicolumn{1}{c|}{$N_{ij}\in\{10,\ldots,15\}$} &
\multicolumn{1}{c}{$N_{ij}\in\{4,5\}$} &
\multicolumn{1}{c}{$N_{ij}\in\{5,\ldots,10\}$} &
\multicolumn{1}{c}{$N_{ij}\in\{10,\ldots,15\}$} \\
\midrule

\multirow{6}{*}{\textbf{$\boldsymbol{J=30}$}}
& DFPCA         & 1.020 & 0.995 & 0.924 & 1.031 & 0.987 & 0.913 & 1.029 & 0.987 & 0.910 \\
& PADA          & 1.371 & 0.837 & 0.521 & 1.188 & 0.845 & 0.537 & 1.279 & 0.817 & 0.494 \\
& LRFPCA        & 1.044 & 0.984 & 0.922 & 1.034 & 0.991 & 0.912 & 1.037 & 0.983 & 0.894 \\
& VMFPCA        & 1.260 & 0.903 & 0.734 & 1.264 & 0.944 & 0.691 & 1.237 & 0.884 & 1.074 \\
& Spectral MPCA & 0.803 & 0.629 & 0.435 & 0.729 & 0.638 & 0.445 & 0.785 & 0.625 & 0.437 \\
\midrule

\multirow{6}{*}{\textbf{$\boldsymbol{J=60}$}}
& DFPCA         & 1.029 & 0.965 & 0.885 & 1.031 & 0.965 & 0.888 & 1.019 & 0.951 & 0.878 \\
& PADA          & 0.843 & 0.552 & 0.416 & 0.848 & 0.539 & 0.408 & 0.829 & 0.541 & 0.405 \\
& LRFPCA        & 0.995 & 0.970 & 0.896 & 0.985 & 0.968 & 0.893 & 0.999 & 0.949 & 0.871 \\
& VMFPCA        & 0.937 & 0.818 & 0.642 & 0.939 & 0.818 & 0.657 & 0.927 & 0.812 & 0.630 \\
& Spectral MPCA & 0.614 & 0.467 & 0.387 & 0.638 & 0.473 & 0.382 & 0.610 & 0.461 & 0.374 \\
\midrule

\multirow{6}{*}{\textbf{$\boldsymbol{J=90}$}}
& DFPCA         & 1.024 & 0.921 & 0.890 & 1.027 & 0.938 & 0.891 & 1.014 & 0.918 & 0.892 \\
& PADA          & 0.742 & 0.457 & 0.384 & 0.738 & 0.445 & 0.381 & 0.731 & 0.440 & 0.381 \\
& LRFPCA        & 0.996 & 0.902 & 0.883 & 0.994 & 0.911 & 0.885 & 0.978 & 0.901 & 0.871 \\
& VMFPCA        & 0.918 & 0.824 & 0.640 & 0.909 & 0.827 & 0.644 & 0.907 & 0.825 & 0.638 \\
& Spectral MPCA & 0.605 & 0.397 & 0.370 & 0.603 & 0.398 & 0.371 & 0.608 & 0.391 & 0.367 \\
\bottomrule
\end{tabular}
}
\end{sidewaystable}

\section{Case Study}\label{section_casestudy}

We analyze hourly PM$_{2.5}$ measurements (particulate matter with diameters below $2.5\,\mu$m) collected at six air-quality monitoring stations in Baoding, China, over six winter weeks in 2013. The resulting dataset contains incomplete daily trajectories over $42$ days and is treated as MFTS data. Following \citet{aue2015prediction, shang2017functional, hormann2022estimating}, we apply a square-root transformation to stabilize the variance and subtract station-specific mean diurnal functions, using only the training days to estimate the means.

\begin{figure}[!ht]
\captionsetup{width=\linewidth}
 \centering
		\includegraphics[width=0.8\linewidth]{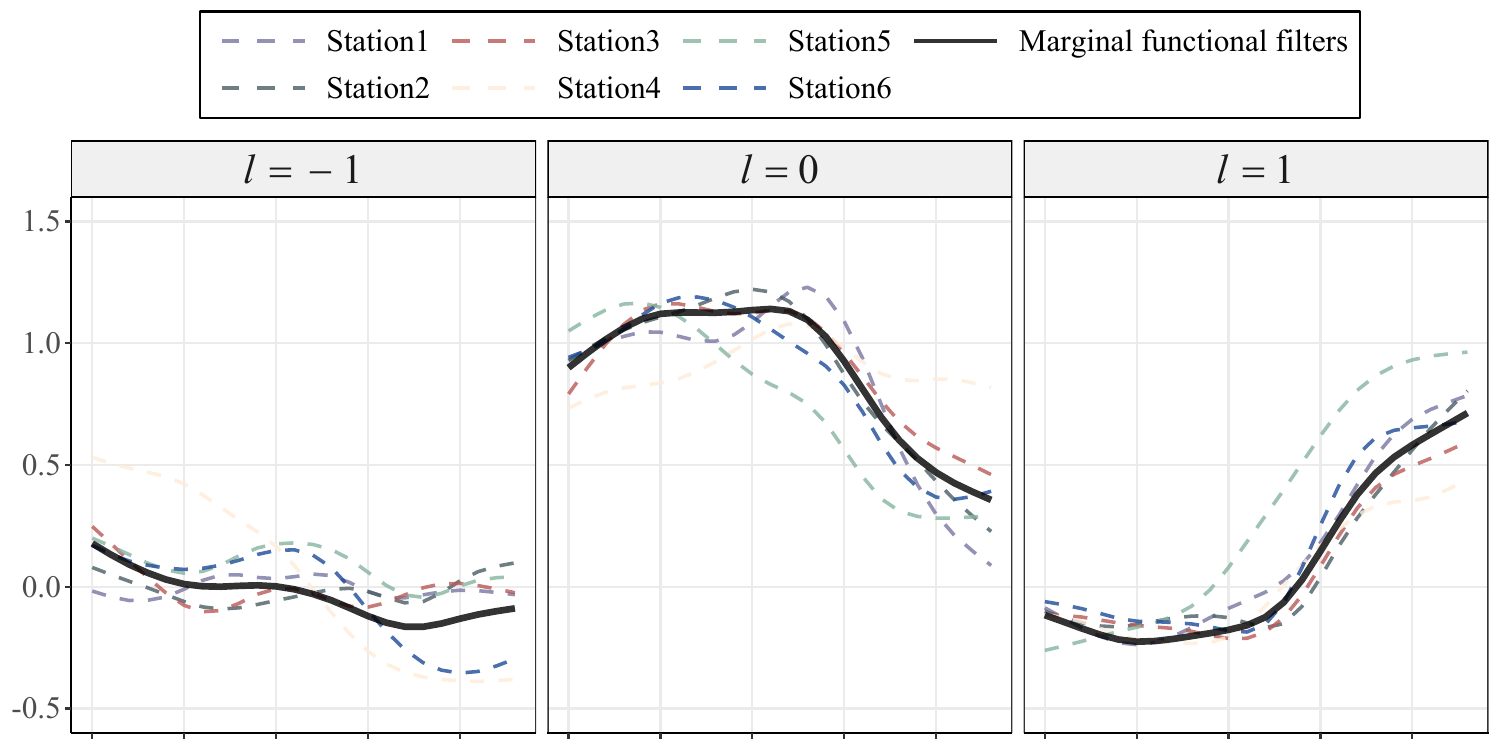}
	\caption{The marginal functional filters (solid lines) and the station-specific functional filters (dashed lines) for the first component.}
 \label{eigenfun1}
\end{figure}

We compare functional filters derived from the PM$_{2.5}$ dataset along two axes: the shared marginal functional filters $\{\hat{\phi}_{1l}(\cdot): |l|\leq L_1\}$ in Spectral MPCA (see Equation~\eqref{marginal_opt_ff_est}) versus the station-specific functional filters $\{\hat{\phi}_{i1l}(\cdot): |l|\leq L_1,\, i\in[p]\}$ computed by PADA \citep{guo2024unified}. The results are shown in Figure~\ref{eigenfun1}. We observe that the marginal functional filters capture common dynamics shared across most stations. This similarity suggests that a substantial portion of the variability can be explained by a shared set of filters, supporting the integration of information across subjects to learn efficient representations for the MFTS.

To evaluate the imputation and forecasting performance of Spectral MPCA, we compute the NMSE and NMSPE defined in Section~\ref{section_simulation} using the discrete observations in this dataset. To this end, we split the data into a training set of $J$ days and a test set of $P$ days. Imputation is performed and evaluated on the training set, while forecasting models are trained on the training set and assessed on the $P$ held-out days. Table~\ref{NMSE_NMSPE_real} reports NMSE for imputation and NMSPE for forecasting across the five methods considered in Section~\ref{section_simulation}. We observe that Spectral MPCA achieves the best forecasting accuracy and remains competitive for imputation, highlighting the superiority of our method for MFTS tasks.

\begin{table}[H]
\centering
\caption{The NMSEs and NMSPEs on Baoding PM$_{2.5}$ data.}
\label{NMSE_NMSPE_real}
\resizebox{0.9\textwidth}{!}{
\begin{tabular}{llccccc}
\toprule
\multicolumn{2}{c}{} 
& \textbf{DFPCA} 
& \textbf{PADA} 
& \textbf{LMFPCA} 
& \textbf{VMFPCA} 
& \textbf{Spectral MPCA} \\
\midrule

\multirow{2}{*}{\textbf{Imputation}}
& $\boldsymbol{J=39,\,P=3}$ 
& 0.378 & 0.267 & 0.543 & 0.544 & 0.288 \\
& $\boldsymbol{J=37,\,P=5}$ 
& 0.394 & 0.263 & 0.538 & 0.538 & 0.284 \\
\midrule

\multirow{2}{*}{\textbf{Forecasting}}
& $\boldsymbol{J=39,\,P=3}$ 
& 1.036 & 1.837 & 0.863 & 0.864 & 0.643 \\
& $\boldsymbol{J=37,\,P=5}$ 
& 0.966 & 1.740 & 0.873 & 0.938 & 0.785 \\
\bottomrule
\end{tabular}
}
\end{table}

\section{Discussion}
In this paper, we propose a frequency-domain integrative framework for MFTS that jointly captures inherent multi-way dependence structures and improves downstream analyses. Based on a novel concept called the marginal spectral density function, our proposed marginal dynamic KL expansion unifies within-curve, temporal, and cross-individual dependencies in the frequency domain. This framework contributes to MFTS analysis in three aspects: first, it ensures statistical efficiency and optimality for MFTS representation and modelling; second, it develops a practical estimation pipeline, including spectral density recovery and Bayes-based score extraction, to handle discretely observed curves and measurement error in a robust manner; third, it delivers a structured multivariate time series representation from the original MFTS, enabling the use of standard multivariate time series tools for modeling and inference with practical MFTS data.

In this article, forecasting with Spectral MPCA is performed using only historical trajectories. In many scenarios, it would be valuable to extend our method to incorporate additional information, such as exogenous covariates or subject-specific spatial features, to further improve predictive accuracy. In addition, extending this framework to high-dimensional regimes is also an important direction \citep{gao2019high, chang2024modelling, hallin2023factor}, as modern applications often involve a large number of subjects for estimation and inference. We leave these generalizations for future work.

\bibliographystyle{apalike}
\bibliography{paper-ref.bib}

\clearpage

\appendix
\section*{\LARGE Appendix}
\addcontentsline{toc}{section}{Appendix}

\section{Propositions and Proofs}\label{SM_TD}
\begin{proposition}
    For each $\omega \in [-\pi,\pi]$, the marginal spectral density function $f_{\mathcal{S}}(\cdot, \cdot \mid \omega)$ defined in \eqref{margin_spec} is a continuous, hermitian, and positive-definite function. 
    \label{marginal_sd_property}
\end{proposition}

\begin{proof}
    The continuity and hermiticity follow from the corresponding properties of individual spectral density functions $\{ f_{ii}(\cdot, \cdot \mid \omega); i \in [p], \omega \in [-\pi,\pi] \}$ \citep{hormann2015dynamic}. Noting that for any $\omega \in [-\pi,\pi]$ and $i \in [p]$, $f_{ii}(\cdot, \cdot \mid \omega)$ is positive-definite, i.e., 
    \begin{align*}
        \int_0^1 \int_0^1  f_{ii}(t, s \mid \omega) u(t) \overline{u(s)}   \mathrm{d} t \mathrm{d} s \geq 0, \quad u(\cdot) \in L^2([0,1], \mathbb{C}).
    \end{align*}
    Therefore, according to Fubini's theorem, we have
    \begin{align*}
        &\int_0^1 \int_0^1  f_{\mathcal{S}}(t, s \mid \omega) u(t) \overline{u(s)}   \mathrm{d} t \mathrm{d} s \\
        &= \frac{1}{p} \sum_{i \in [p] } \int_0^1 \int_0^1  f_{ii}(t, s \mid \omega) u(t) \overline{u(s)}   \mathrm{d} t \mathrm{d} s \geq 0,  \quad \omega \in [-\pi,\pi],
    \end{align*}
    for any $u(\cdot)$.
\end{proof}

\

\

\begin{proposition}
Under the marginal dynamic KL expansion \eqref{marginal_dfpca}, the sequence $\boldsymbol{\xi}_{i \cdot k} = \{ \xi_{ijk}; j \in \mathbb{Z} \}$ form a weakly stationary time series with spectral density 
\begin{align*}
    \tilde{\eta}_{ik}(\omega) := \int_0^1 \int_0^1 \psi_k(t \mid \omega) f_{ii}(t,s \mid  \omega) \overline{\psi_k(s \mid \omega)}  \mathrm{d}t  \mathrm{d}s, \quad \omega \in [-\pi,\pi],
\end{align*}
where $i \in [p]$ and $k \geq 1$.
\label{marginal_score_sd}
\end{proposition}

\begin{proof}
For given $i \in [p]$ and $k \geq 1$, let $C_{h}^{\boldsymbol{\xi}_{i j k}} := \operatorname{E} \xi_{i(j+h)k} \xi_{ijk}$ with $j,h \in \mathbb{Z}$.
According to the definition of scores $\xi_{ijk}$ in the main text, we have
\begin{equation*}
    \begin{aligned}
            C_{h}^{\boldsymbol{\xi}_{i j k}} 
    &= \operatorname{E} \bigg\{ \sum_{l_1 \in \mathbb{Z}} \langle \varepsilon_{i(j+h-l_1)}, \phi_{kl_1} \rangle \bigg\} \bigg\{ \sum_{l_2 \in \mathbb{Z}} \langle \varepsilon_{i(j-l_2)}, \phi_{kl_2} \rangle \bigg\} \\
    &= \sum_{l_1, l_2 \in \mathbb{Z}} \int_0^1 \int_0^1 \phi_{kl_1}(t) \operatorname{E} \{\varepsilon_{i(j+h-l_1)}, \varepsilon_{i(j-l_2)}\} \phi_{kl_2}(s)  \mathrm{d}t \mathrm{d}s,
    \end{aligned}
\end{equation*}
where $\phi_{kl}(\cdot)$ represents the marginal functional filters in \eqref{marginal_dfpca}. Since $\boldsymbol{\varepsilon}_{i\cdot} = \{\varepsilon_{ij}; j \in \mathbb{Z}\}$ is a weakly stationary functional time series with auto-covariance $\{C_{iih}(\cdot,\cdot); h \in \mathbb{Z}\}$, we have
\begin{equation}
    \begin{aligned}
    C_{h}^{\boldsymbol{\xi}_{i j k}} 
    = \sum_{l_1, l_2 \in \mathbb{Z}} \int_0^1 \int_0^1 \phi_{kl_1}(t) C_{ii(l_2-l_1+h)}(t,s) \phi_{kl_2}(s)  \mathrm{d}t \mathrm{d}s.
    \end{aligned}
    \label{auto_cov_score}
\end{equation}
Therefore, given $i$ and $k$, $C_{h}^{\boldsymbol{\xi}_{i j k}}$ is free of $j$ for any $h \in \mathbb{Z}$, and the weak stationarity is proven. In the following, we simplify $C_{h}^{\boldsymbol{\xi}_{i j k}}$ as $C_{ikh}$ for convenience.

Combining \eqref{auto_cov_score} with the Fourier inversion theorem, the spectral density of scores $\boldsymbol{\xi}_{i \cdot k}$ can be shown as
\begin{align*}
    \tilde{\eta}_{ik}(\omega) 
    & = \frac{1}{2\pi} \sum_{h \in \mathbb{Z}} C_{ikh} \exp(\mathrm{i} h \omega)\\
    &= \frac{1}{2\pi} \sum_{h \in \mathbb{Z}} \sum_{l_1, l_2 \in \mathbb{Z}} \int_0^1 \int_0^1 \phi_{kl_1}(t) C_{ii(l_2-l_1+h)}(t,s) \phi_{kl_2}(s)  \mathrm{d}t \mathrm{d}s \exp(\mathrm{i} h \omega) \\
    &= \frac{1}{2\pi} \sum_{h \in \mathbb{Z}} \sum_{l_1, l_2 \in \mathbb{Z}} \int_0^1 \int_0^1 \int_{-\pi}^{\pi} \phi_{kl_1}(t) f_{ii}(t,s \mid \omega_1) \exp\{-\mathrm{i} (l_2 - l_1 + h) \omega_1\} \phi_{kl_2}(s) \mathrm{d}\omega_1  \mathrm{d}t \mathrm{d}s \exp(\mathrm{i} h \omega) \\
    &= \frac{1}{2\pi} \sum_{h \in \mathbb{Z}} \int_0^1 \int_0^1 \int_{-\pi}^{\pi} \psi_{k}(t \mid \omega_1) f_{ii}(t,s \mid \omega_1) \overline{\psi_{k}(s \mid \omega_1)} \exp(-\mathrm{i} h \omega_1) \mathrm{d}\omega_1  \mathrm{d}t \mathrm{d}s \exp(\mathrm{i} h \omega) \\
    &= \int_0^1 \int_0^1  \psi_{k}(t \mid \omega) f_{ii}(t,s \mid \omega) \overline{\psi_{k}(s \mid \omega)} \mathrm{d}t \mathrm{d}s, \quad \omega \in [-\pi,\pi],\ i \in [p], \ k \geq 1.
\end{align*}
\end{proof}

\

\begin{proof}[Proof for Theorem \ref{opt_approximation}]
For convenience, we introduce some notions about linear operators. The Hilbert-Schmidt norm and the trace norm for an operator $\mathcal{T}: L^{2}([0,1],\mathbb{C}) \to L^{2}([0,1],\mathbb{C})$ are defined as
\begin{align*}
    \|\mathcal{T}\|_{\operatorname{HS}} = \sqrt{\sum_{k=1}^{\infty} \|\mathcal{T}e_k\|^{2}}, \quad \operatorname{tr}(\mathcal{T}) = \sum_{k=1}^{\infty} \langle(\mathcal{T}^{*}\mathcal{T})^{1/2} e_k, e_k \rangle,
\end{align*}
where $e_k(\cdot)$s are orthonormal bases of $L^{2}([0,1],\mathbb{C})$, and $(\cdot)^{*}$ denotes the adjoint operation for an operator. Furthermore, based on the definition of the marginal spectral density function $f_{\mathcal{S}}(\cdot, \cdot \mid \omega)$ in \eqref{margin_spec}, we define the corresponding integral operator as
\begin{align}
    \mathcal{F}_{\mathcal{S}}^{\omega} u := \int_0^1 f_{\mathcal{S}}(t, s \mid \omega) u(t) \mathrm{d}t, 
    \label{F_operator_def}
\end{align}
for any $u(\cdot) \in L^{2}([0,1],\mathbb{C})$ and $\omega \in [-\pi,\pi]$. According to Proposition \ref{marginal_sd_property}, $\mathcal{F}_{\mathcal{S}}^{\omega}$ admits a spectral decomposition \citep{hsing2015theoretical} as
\begin{align*}
    \mathcal{F}_{\mathcal{S}}^{\omega} = \sum_{k=1}^{\infty} \eta_k(\omega) \psi_k(\omega) \otimes \psi_k(\omega), 
\end{align*}
where $\psi_k(\omega) \otimes \psi_k(\omega)$ is a tensor product operator between $L^{2}([0,1],\mathbb{C})$ satisfying
\begin{align*}
    \{\psi_k(\omega) \otimes \psi_k(\omega)\} u = \psi_k(\omega) \langle \psi_k(\omega) , u \rangle, 
\end{align*}
where $u(\cdot) \in L^{2}([0,1],\mathbb{C})$

According to the definition of $\tilde{\boldsymbol{\varepsilon}}_{j}^{K}(\cdot)$ in the main text, we have
\begin{equation}
\begin{aligned}
&\operatorname{E} \| \boldsymbol{\varepsilon}_{j} - \tilde{\boldsymbol{\varepsilon}}^{K}_{j} \|_{p}^{2} \\
&= \operatorname{E} \sum_{i \in [p]} \bigg[ \int_{0}^{1} \big\{
\varepsilon_{ij}(t)
-\sum_{k \in [K]} \sum_{l \in \mathbb{Z}}
\tilde{\phi}_{kl}(t)  \tilde{\xi}_{i(j+l)k} \big\}^2 \mathrm{d}t \bigg] \\
&= \sum_{i \in [p]} \Bigg( \int_{0}^{1} \bigg[ C_{ii0}(t,t) - 2\sum_{k \in [K]} \sum_{l \in \mathbb{Z}} \operatorname{E}\varepsilon_{ij}(t) \tilde{\phi}_{kl}(t) \tilde{\xi}_{i(j+l)k} + \operatorname{E} \bigg\{
\sum_{k \in [K]} \sum_{l \in \mathbb{Z}} \tilde{\phi}_{kl}(t) \tilde{\xi}_{i(j+l)k}
\bigg\}^{2} \bigg] \mathrm{d}t\Bigg), \\
&= \int_{0}^{1} \sum_{i \in [p]} C_{ii0}(t,t) \mathrm{d}t - 2 \int_{0}^{1} \sum_{i \in [p]} \sum_{k \in [K]} \sum_{l \in \mathbb{Z}} \operatorname{E}\varepsilon_{ij}(t) \tilde{\phi}_{kl}(t) \tilde{\xi}_{i(j+l)k} \mathrm{d}t \\
& \quad + \int_{0}^{1} \sum_{i \in [p]} \operatorname{E} \bigg\{
\sum_{k \in [K]} \sum_{l \in \mathbb{Z}} \tilde{\phi}_{kl}(t) \tilde{\xi}_{i(j+l)k}
\bigg\}^{2} \mathrm{d}t,
\end{aligned}
\label{all_eq}
\end{equation}
for any $K \geq 1$.

First, according to the inverse Fourier transform \eqref{r to f}, we have
\begin{align*}
    C_{ii0}(t,t) = \int_{-\pi}^{\pi} f_{ii}(t,t \mid \omega)  \mathrm{d} \omega, \quad t \in [0,1],\ i \in [p].
\end{align*}
Therefore, the first term in \eqref{all_eq} can be shown as
\begin{align}
     \int_{0}^{1} \sum_{i \in [p]} C_{ii0}(t,t) \mathrm{d} t = p \int_{-\pi}^{\pi} \int_{0}^{1} f_{\mathcal{S}}(t,t \mid \omega) \mathrm{d} t  \mathrm{d} \omega = p \int_{-\pi}^{\pi} \operatorname{tr} (\mathcal{F}_{\mathcal{S}}^{\omega})  \mathrm{d} \omega,
    \label{1part_operator_eq}
\end{align}
where $\mathcal{F}_{\mathcal{S}}^{\omega}$ is defined as \eqref{F_operator_def}.

Next, based on the definition of $\tilde{\xi}_{ijk}$ in the main text, we obtain
\begin{align*}
    &\sum_{k \in [K]} \sum_{l \in \mathbb{Z}} \operatorname{E} \varepsilon_{ij}(t) \tilde{\phi}_{kl}(t) \tilde{\xi}_{i(j+l)k} \\
    &= \sum_{k \in [K]} \sum_{l_1, l_2 \in \mathbb{Z}} \tilde{\phi}_{kl_1}(t) \int_{0}^{1} C_{ii(l_1-l_2)}(s,t) \tilde{v}_{kl_2}(s) \mathrm{d} s \\
    &= \sum_{k \in [K]} \sum_{l_1, l_2 \in \mathbb{Z}} \tilde{\phi}_{kl_1}(t) \int_{-\pi}^{\pi}  \int_{0}^{1} f_{ii}(s,t \mid \omega) \exp\{ \mathrm{i}(l_2-l_1) \omega\} \tilde{v}_{kl_2}(s) \mathrm{d} s \mathrm{d} \omega, \quad t \in [0,1], \ i \in [p].
\end{align*}
where $\{\tilde{\phi}_{kl}(\cdot); l \in \mathbb{Z}\}$ and $\{\tilde{v}_{kl}(\cdot); l \in \mathbb{Z}\}$ are sequences belonging to $\mathcal{H}([0,1], \mathbb{R})$. According to the definition of $\mathcal{H}([0,1], \mathbb{R})$ in the main text, there exist $\tilde{\psi}_{k}(\cdot \mid \omega)$ and $\tilde{w}_{k}(\cdot \mid \omega)$ s.t. $\tilde{\psi}_{k}(\cdot \mid \omega) = \sum_{l \in \mathbb{Z}} \tilde{\phi}_{kl}(\cdot) \exp(\mathrm{i} l \omega)$ and $\tilde{w}_{k}(\cdot \mid \omega) = \sum_{l \in \mathbb{Z}} \tilde{v}_{kl}(\cdot) \exp(\mathrm{i} l \omega)$ for all $\omega \in [-\pi,\pi]$ and $k \in [K]$. Therefore, we have
\begin{equation}
    \begin{aligned}
            &\sum_{k \in [K]} \sum_{l \in \mathbb{Z}} \operatorname{E} \varepsilon_{ij}(t) \tilde{\phi}_{kl}(t) \tilde{\xi}_{i(j+l)k} \\
    &= \int_{-\pi}^{\pi}  \int_{0}^{1} \sum_{k \in [K]} \overline{\tilde{\psi}_{k}(t \mid \omega)} f_{ii}(s,t \mid \omega) \tilde{w}_{k}(s \mid \omega) \mathrm{d} s \mathrm{d} \omega, \quad t \in [0,1], \ i \in [p].
    \end{aligned}
    \label{2part_subeq}
\end{equation}
According to \eqref{2part_subeq}, the second term in \eqref{all_eq} can be shown as
\begin{align*}
    &2 \int_{0}^{1} \sum_{i \in [p]} \sum_{k \in [K]} \sum_{l \in \mathbb{Z}} \operatorname{E}\varepsilon_{ij}(t) \tilde{\phi}_{kl}(t) \tilde{\xi}_{i(j+l)k} \mathrm{d} t \\
    &= p \int_{-\pi}^{\pi} \int_{0}^{1} \int_{0}^{1} f_{\mathcal{S}}(s,t \mid \omega) G_K(t,s \mid \omega) + G_K(t,s \mid \omega) f_{\mathcal{S}}(s,t \mid \omega)   \mathrm{d} s \mathrm{d} t \mathrm{d} \omega, \quad t \in [0,1],
\end{align*}
where $G_K(t,s \mid \omega) = \sum_{k \in [K]} \overline{\tilde{\psi}_{k}(t \mid \omega)} \tilde{w}_{k}(s \mid \omega)$. Furthermore, we define $\mathcal{G}_{K}^{\omega}$ as the corresponding integral operator of $G_K(\cdot, \cdot \mid \omega)$ for given $K \geq 1$ and $\omega \in [-\pi,\pi]$, then we have
\begin{align}
    2 \int_{0}^{1} \sum_{i \in [p]} \sum_{k \in [K]} \sum_{l \in \mathbb{Z}} \operatorname{E}\varepsilon_{ij}(t) \tilde{\phi}_{kl}(t) \tilde{\xi}_{i(j+l)k} \mathrm{d} t 
    = p \int_{-\pi}^{\pi} \operatorname{tr} \big\{ \mathcal{F}_{\mathcal{S}}^{\omega} \circ \mathcal{G}_{K}^{\omega} + (\mathcal{G}_{K}^{\omega})^{*} \circ (\mathcal{F}_{\mathcal{S}}^{\omega})^{*} \big\}  \mathrm{d} \omega,
    \label{2part_operator_eq}
\end{align}
where $\circ$ represents the multiplication between operators.

Similar to the above derivation, it can be proved that
\begin{align*}
&\operatorname{E} \bigg\{ \sum_{k \in [K]} \sum_{l \in \mathbb{Z}} \tilde{\phi}_{kl}(t) \tilde{\xi}_{i(j+l)k} \bigg\}^{2} \\
&= \sum_{k_1, k_2 \in [K]} \sum_{l_1, l_2 \in \mathbb{Z}} \tilde{\phi}_{k_1l_1}(t) \tilde{\phi}_{k_2l_2}(t) \operatorname{E} \big\{\tilde{\xi}_{i(j+l_1)k_1} \tilde{\xi}_{i(j+l_2)k_2} \big\} \\
&=\sum_{k_1, k_2 \in [K]} \sum_{l_1, l_2, l_3, l_4 \in \mathbb{Z}} \tilde{\phi}_{k_1l_1}(t) \tilde{\phi}_{k_2l_2}(t) \int_{0}^{1} \int_{0}^{1} C_{ii(l_1-l_3-l_2+l_4)}(s,u) \tilde{v}_{k_1l_3}(s) \tilde{v}_{k_2l_4}(u) \mathrm{d} s \mathrm{d} u \\
&=\sum_{k_1, k_2 \in [K]} \sum_{l_1, l_2, l_3, l_4 \in \mathbb{Z}} \tilde{\phi}_{k_1l_1}(t) \tilde{\phi}_{k_2l_2}(t) \\
& \quad \cdot \int_{-\pi}^{\pi} \int_{0}^{1} \int_{0}^{1} f_{ii}(s,u \mid \omega) \exp\{ \mathrm{i} (l_2+l_3-l_1-l_4)\omega \} \tilde{v}_{k_1l_3}(s) \tilde{v}_{k_2l_4}(u) \mathrm{d} s \mathrm{d} u \mathrm{d}\omega \\
&= \int_{-\pi}^{\pi} \int_{0}^{1} \int_{0}^{1} \sum_{k_1, k_2 \in [K]} \overline{\tilde{\psi}_{k_1}(t \mid \omega)}  {\tilde{\psi}_{k_2}(t \mid \omega)} f_{ii}(s,u \mid \omega) \tilde{w}_{k_1}(s \mid \omega) \overline{\tilde{w}_{k_2}(u \mid \omega)} \mathrm{d} s \mathrm{d} u \mathrm{d}\omega, \quad t \in [0,1],
\end{align*}
for any given $i \in [p]$. According to the definition of operators $\mathcal{G}_{K}^{\omega}$ and $\mathcal{F}_{\mathcal{S}}^{\omega}$, the third term in \eqref{all_eq} can be shown as
\begin{equation}
\begin{aligned}
&\int_{0}^{1} \sum_{i \in [p]} \operatorname{E} \bigg\{ \sum_{k \in [K]} \sum_{l \in \mathbb{Z}} \tilde{\phi}_{kl}(t) \tilde{\xi}_{i(j+l)k} \bigg\}^{2} \mathrm{d} t \\
&= p \int_{-\pi}^{\pi} \int_{0}^{1} \int_{0}^{1} \int_{0}^{1} \overline{G_K(t,u \mid \omega)}  f_{\mathcal{S}}(s,u \mid \omega) G_K(t,s \mid \omega)  \mathrm{d} t \mathrm{d} s \mathrm{d} u \mathrm{d}\omega \\
&= p \int_{-\pi}^{\pi}  \operatorname{tr}\big\{ (\mathcal{G}_{K}^{\omega})^{*} \circ \mathcal{F}_{\mathcal{S}}^{\omega} \circ \mathcal{G}_{K}^{\omega} \big\}  \mathrm{d}\omega.
\end{aligned}
 \label{3part_operator_eq}
\end{equation}

Combining \eqref{all_eq}, \eqref{1part_operator_eq}, \eqref{2part_operator_eq} and \eqref{3part_operator_eq}, we have
\begin{align*}
    \operatorname{E} \| \boldsymbol{\varepsilon}_{j} - \tilde{\boldsymbol{\varepsilon}}^{K}_{j} \|_{p}^{2} = p \int_{-\pi}^{\pi} \| (\mathcal{F}_{\mathcal{S}}^{\omega})^{1/2} - (\mathcal{F}_{\mathcal{S}}^{\omega})^{1/2} \circ \mathcal{G}_{K}^{\omega} \|_{\operatorname{HS}}  \mathrm{d}\omega,
\end{align*}
where $(\mathcal{F}_{\mathcal{S}}^{\omega})^{1/2} = \sum_{k=1}^{\infty} \{\eta_k(\omega)\}^{1/2} \psi_k(\omega) \otimes \psi_k(\omega)$. Therefore, $\operatorname{E} \| \boldsymbol{\varepsilon}_{j} - \tilde{\boldsymbol{\varepsilon}}^{K}_{j} \|_{p}^{2}$ reaches its minimum $\operatorname{E} \| \boldsymbol{\varepsilon}_{j} - {\boldsymbol{\varepsilon}}^{K}_{j} \|_{p}^{2}$ when $\tilde{\psi}_{k}(t \mid \omega) = \tilde{w}_k(t \mid \omega) =  {\psi}_{k}(t \mid \omega) $ for $t \in [0,1]$, $\omega \in [-\pi,\pi]$ and $k \in [K]$ \citep{tan2024graphical}. 
\end{proof}

\newpage
\clearpage

\section{Implementation Details}

\subsection{Matrix Forms for the Score Extraction Optimization}\label{mat_posterior}
In this subsection, we provide detailed explanations of the matrix forms \eqref{mat_log_likelihood} and \eqref{mat_log_prior}. For this purpose, we define $\tilde{N} := \sum_{i \in [p]} \sum_{j \in [J]} N_{ij}$ as the total number of observations, and introduce a one-to-one mapping
\begin{align*}
    m : \{(i,j,z) : i \in [p], j \in [J], z \in [N_{ij}]\} \to [\tilde{N}]
\end{align*}
that orders all triplets $(i,j,z)$. Using this mapping, the demeaned observations $\{\tilde{Y}_{ijz} = Y_{ijz} - \mu_{i}(t_{ijz}); i \in [p], j \in [J], z \in [N_{ij}] \}$ can be vectorized as
\begin{align}
\bm y := \{\tilde{Y}_{m(i,j,z)}\} = (\tilde{\bm Y}^\top_{11 \cdot}, \ldots, \tilde{\bm Y}^\top_{1 J \cdot}, \ldots, \tilde{\bm Y}^\top_{p 1 \cdot}, \ldots, \tilde{\bm Y}^\top_{p J \cdot})^\top \in \mathbb{R}^{\tilde{N}},
\label{obs_stack}
\end{align}
where $\tilde{\bm Y}_{i j \cdot} = (\tilde{Y}_{i j 1}, \ldots, \tilde{Y}_{i j N_{ij}})^\top \in \mathbb{R}^{N_{ij}}$ for any $i$ and $j$. Similarly, the scores $\{\boldsymbol{\xi}_{\cdot \cdot k}; k \in [K]\}$ in the Bayesian formulation \eqref{posterior} can be vectorized via the mapping 
\begin{align*}
    c : \{(i',j',k) : i' \in [p], j' \in \{1 - L_k, \dots, J + L_k\}, k \in [K]\} \to [d_\xi],
\end{align*}
where $d_\xi := \sum_{k \in [K]} p T_k$ and $T_k = J + 2L_k$ for any $k$. Under this mapping, we define the vectorized scores as
\begin{align}
    \bm \xi := \{\xi_{c(i',j',k)}\} = (\bm{\xi}_{1}^\top,\dots,\bm{\xi}_{K}^\top)^\top \in \mathbb{R}^{d_\xi},
    \label{score_stack}
\end{align}
where 
\begin{equation*}
  \bm{\xi}_{k}
  :=
  \bigl(
    \xi_{1(1-L_k)k},
    \dots,
    \xi_{p(1-L_k)k},
    \dots,
    \xi_{1(J+L_k)k},
    \dots,
    \xi_{p(J+L_k)k}
  \bigr)^\top
  \in \mathbb{R}^{p T_k}, \quad k \in [K].
\end{equation*}

To obtain the matrix form of the log-likelihood $\log \mathcal{L}(\boldsymbol{\xi}_{\cdot \cdot 1}, \ldots, \boldsymbol{\xi}_{\cdot \cdot K} \mid \boldsymbol{Y})$, we define the $(m(i,j,z), c(i',j',k))$-th element of the design matrix $\bm A \in \mathbb{R}^{\tilde{N} \times d_\xi}$ as
\begin{equation}
  \label{eq:A-entry}
  [\bm A]_{m(i,j,z), c(i',j',k)}
  :=
  \begin{cases}
    \phi_{kl}(t_{ijz}),
    & \text{if } i' = i,\ j' = j + l \text{ for some } l \in \{-L_k,\dots,L_k\},\\[0.2em]
    0, & \text{otherwise}.
  \end{cases}
\end{equation}
Combining \eqref{score_stack} and \eqref{eq:A-entry}, the $m(i,j,z)$-th element of $\bm{A\xi} \in \mathbb{R}^{\tilde{N}}$ is
\begin{equation}
  \label{eq:A-predictor}
  [\bm{A \xi}]_{m(i,j,z)}
  =
  \sum_{k=1}^K \sum_{|l| \le L_k} \phi_{kl}(t_{ijz}) \, \xi_{i (j+l) k}.
\end{equation}
Moreover, we define the weight matrix $\bm W \in \mathbb{R}^{\tilde{N} \times \tilde{N}}$ by setting its diaonal element as 
\begin{equation}
  \label{eq:W-def}
  [\bm W]_{\tilde{n}, \tilde{n}} := \sigma_{i}^{-2}, \quad \tilde{n} \in [\tilde{N}],
\end{equation}
if $\tilde{n} = m(i,j,k)$ for any $j$ and $k$, so that each observation from subject $i$ is weighted by $\sigma_i^{-2}$. According to \eqref{obs_stack}, \eqref{eq:A-predictor} and \eqref{eq:W-def}, we have
\begin{equation}
    \begin{aligned}
          \log \mathcal{L}(\boldsymbol{\xi}_{\cdot \cdot 1}, \ldots, \boldsymbol{\xi}_{\cdot \cdot K} \mid \boldsymbol{Y}) &=
  -\frac12
  \sum_{i \in [p]} \sum_{j \in [J]} \sum_{z \in [N_{ij}]}
  \sigma_i^{-2}
  \bigl\{
    \tilde{Y}_{ijz}
    -
    \sum_{k=1}^K \sum_{|l| \le L_k} \phi_{kl}(t_{ijz}) \xi_{i (j+l) k}
  \bigr\}^2 + C \\
  &= -\frac12 (\bm y - \bm{A \xi})^\top \bm W (\bm y - \bm{A \xi}) + C \\
  &=
  -\frac12 \bigl\| \bm{W}^{1/2} (\bm y - \bm{A \xi}) \bigr\|_{\text{E}}^2 + C,
    \end{aligned}
    \label{eq:loglik-sum}
\end{equation} 
where $\| \cdot \|_{\text{E}}$ is the Euclidean norm for vectors, and $C$ is a constant that is independent of $(\boldsymbol{\xi}_{\cdot \cdot 1}, \ldots, \boldsymbol{\xi}_{\cdot \cdot K})$.

To obtain the matrix form of the joint log-prior $\sum_{k \in [K]}\log \pi(\boldsymbol{\xi}_{\cdot \cdot k}) $, we introduce the discrete Fourier transform matrix $\bm F_k \in \mathbb{C}^{J \times T_k}$ with $(j,r)$-th entries
\[
  [\bm F_k]_{j,r} = \frac{1}{\sqrt{2\pi T_k}} \exp(\mathrm{i} r \omega_j),
  \qquad j \in [J], \ r \in [T_k], \ \omega_j \in S_{J},
\]
for any $k \in [K]$. The vectorized discrete Fourier transform of the scores can be written as
\begin{align}
    \tilde{\bm \xi}_k := (\tilde{\xi}_{1k}(\omega_1), \ldots, \tilde{\xi}_{pk}(\omega_1), \ldots, \tilde{\xi}_{1k}(\omega_J), \ldots, \tilde{\xi}_{pk}(\omega_J))^{\top} = (\bm F_k \otimes \bm I_p) \bm{\xi}_k, \quad k \in [K],
    \label{vec_DFT}
\end{align}
where $\otimes$ denotes the Kronecker product and $\bm I_p$ is the $p \times p$ identity matrix. Moreover, we define the block-diagonal matrix
\begin{equation}
  \label{eq:Dk-def}
  \bm D_k := \bigoplus_{\omega_j \in S_{J}} \bm \Phi_k(\omega_j) \in \mathbb{C}^{p J \times p J}, \quad k \in [K],
\end{equation}
whose diagonal blocks are the $p \times p$ matrices $\{\bm \Phi_k(\omega_j); j \in [J]\}$ defined in the main text. Combining \eqref{vec_DFT} and \eqref{eq:Dk-def}, we have
\begin{equation}
    \begin{aligned}
    &\sum_{j \in [J]}
  \left\{ \tilde{\boldsymbol{\xi}}_{\cdot k}(\omega_j) \right\}^* \boldsymbol{\Phi}_k(\omega_j) \tilde{\boldsymbol{\xi}}_{\cdot k}(\omega_j) \\
  &=
  \tilde{\bm \xi}_{k}^* \bm D_k \tilde{\bm \xi}_{k} \\
  &=
  \bm \xi_{k}^*(\bm F_k \otimes \bm I_p)^* \bm D_k (\bm F_k \otimes \bm I_p) \bm \xi_{k}, \quad k \in [K].
    \end{aligned}
    \label{eq:prior-quad-step}
\end{equation}
According to \eqref{eq:Dk-def} and \eqref{eq:prior-quad-step}, we define the prior matrix $\bm Q$ as
\begin{align}
\bm Q_k=(\bm F_k\!\otimes\! \bm I_p)^{*}\!\left(\bigoplus_{\omega_j\in\mathcal{S}_J}\boldsymbol{\Phi}_k(\omega_j)\right)\!(\bm F_k\!\otimes\! \bm I_p),
\qquad
\bm Q=\mathrm{diag}(\bm Q_1,\ldots,\bm Q_K),
\label{Q_def_SM}
\end{align}
and the joint log-prior can be shown as
\begin{equation}
    \begin{aligned}
        \sum_{k \in [K]}\log \pi(\boldsymbol{\xi}_{\cdot \cdot k}) 
        &= -\frac{1}{2} \sum_{k \in [K]} \sum_{j \in [J]} 
\left\{ \tilde{\boldsymbol{\xi}}_{\cdot k}(\omega_j) \right\}^* \boldsymbol{\Phi}_k(\omega_j) \tilde{\boldsymbol{\xi}}_{\cdot k}(\omega_j) + C \\
&= -\frac{1}{2} \sum_{k \in [K]} {\bm \xi}_{k}^* \bm Q_k {\bm \xi}_{k} + C \\
&= -\frac{1}{2} \bm\xi^* \bm Q {\bm\xi} +C.
    \end{aligned}
    \label{eq:logprior-sum}
\end{equation}
Combining \eqref{eq:loglik-sum}, \eqref{eq:logprior-sum} and \eqref{posterior}, the MAP estimator for the scores is obtained by solving
\[
\arg\min_{\bm\xi} \frac{1}{2}\,\big\|\bm W^{1/2}\big(\bm y-\bm{A\xi}\big)\big\|_{\text{E}}^2\;+\;\frac{1}{2}\,\bm\xi^{*}\bm Q\, \bm\xi,
\]
as described in Algorithm \ref{alg:mfts_estimation}.

\textbf{Example}:  To make the above notation concrete, we consider a small illustrative example with $(p, J, K, L_1) = (2,3,1,1)$ and set $N_{ij} = 2$ for all $i$ and $j$. In this case, $T_1 = J + 2L_1 = 5$ and the total number of observations is $\tilde{N} = p J N_{ij} = 2 \times 3 \times 2 = 12$. 

In this example, we first explicitly display $\bm y$, $\bm \xi$, $\bm A$, and $\bm W$ corresponding to the matrix form of the log-likelihood $-\frac12 \bigl\| \bm W^{1/2} (\bm y - \bm{A \xi}) \bigr\|_{\text{E}}^2$. The observation vector $\bm y$ is
\begin{align*}
    \bm y = \bigl(
    \tilde{Y}_{111},
    \tilde{Y}_{112},
    \tilde{Y}_{121},
    \tilde{Y}_{122},
    \tilde{Y}_{131},
    \tilde{Y}_{132}, 
    \tilde{Y}_{211},
    \tilde{Y}_{212},
    \tilde{Y}_{221},
    \tilde{Y}_{222},
    \tilde{Y}_{231},
    \tilde{Y}_{232}
  \bigr)^\top \in \mathbb{R}^{12}.
\end{align*}
The score vector $\bm\xi$ is 
\begin{align*}
      \bm\xi =
  \bigl(
    \xi_{1,0,1},
    \xi_{2,0,1},
    \xi_{1,1,1},
    \xi_{2,1,1},
    \xi_{1,2,1},
    \xi_{2,2,1},
    \xi_{1,3,1},
    \xi_{2,3,1},
    \xi_{1,4,1},
    \xi_{2,4,1}
  \bigr)^\top \in \mathbb{R}^{10}.
\end{align*}
The weight matrix $\bm W$ is
\begin{equation*}
  \bm W
  =
  \operatorname{diag}\bigl(
    \underbrace{\sigma_1^{-2}, \dots, \sigma_1^{-2}}_{\text{6 entries}},
    \underbrace{\sigma_2^{-2}, \dots, \sigma_2^{-2}}_{\text{6 entries}}
  \bigr).
\end{equation*}
The design matrix $\bm A$ is
\begin{equation*}
\resizebox{\textwidth}{!}{$
\bm A = 
\begin{pmatrix}
\phi_{1(-1)}(t_{111}) & 0 & \phi_{10}(t_{111}) & 0 & \phi_{11}(t_{111}) & 0 & 0 & 0 & 0 & 0 \\
\phi_{1(-1)}(t_{112}) & 0 & \phi_{10}(t_{112}) & 0 & \phi_{11}(t_{112}) & 0 & 0 & 0 & 0 & 0 \\
0 & 0 & \phi_{1(-1)}(t_{121}) & 0 & \phi_{10}(t_{121}) & 0 & \phi_{11}(t_{121}) & 0 & 0 & 0 \\
0 & 0 & \phi_{1(-1)}(t_{122}) & 0 & \phi_{10}(t_{122}) & 0 & \phi_{11}(t_{122}) & 0 & 0 & 0 \\
0 & 0 & 0 & 0 & \phi_{1(-1)}(t_{131}) & 0 & \phi_{10}(t_{131}) & 0 & \phi_{11}(t_{131}) & 0 \\
0 & 0 & 0 & 0 & \phi_{1(-1)}(t_{132}) & 0 & \phi_{10}(t_{132}) & 0 & \phi_{11}(t_{132}) & 0 \\
0 & \phi_{1(-1)}(t_{211}) & 0 & \phi_{10}(t_{211}) & 0 & \phi_{11}(t_{211}) & 0 & 0 & 0 & 0 \\
0 & \phi_{1(-1)}(t_{212}) & 0 & \phi_{10}(t_{212}) & 0 & \phi_{11}(t_{212}) & 0 & 0 & 0 & 0 \\
0 & 0 & 0 & \phi_{1(-1)}(t_{221}) & 0 & \phi_{10}(t_{221}) & 0 & \phi_{11}(t_{221}) & 0 & 0 \\
0 & 0 & 0 & \phi_{1(-1)}(t_{222}) & 0 & \phi_{10}(t_{222}) & 0 & \phi_{11}(t_{222}) & 0 & 0 \\
0 & 0 & 0 & 0 & 0 & \phi_{1(-1)}(t_{231}) & 0 & \phi_{10}(t_{231}) & 0 & \phi_{11}(t_{231}) \\
0 & 0 & 0 & 0 & 0 & \phi_{1(-1)}(t_{232}) & 0 & \phi_{10}(t_{232}) & 0 & \phi_{11}(t_{232})
\end{pmatrix}.
$}
\end{equation*}

For the matrix form of log-prior $-\frac{1}{2} \bm\xi^* \bm Q \bm{\xi}$, we display $\bm Q$ explicitly. The discrete Fourier transform matrix $\bm F_1$ has entries $[\bm F_1]_{j,r} = \frac{1}{\sqrt{2\pi T_1}} \exp(\mathrm{i} r \omega_j)$. For each $\omega_j$, the matrix $\bm\Phi_1(\omega_j)$ is a $2 \times 2$ diagonal matrix with elements $\bigl(\tilde{\eta}_{11}^{-1}(\omega_j),\tilde{\eta}_{21}^{-1}(\omega_j)\bigr)$, where $\tilde{\eta}_{ik}(\cdot)$ is defined in \eqref{margin_score_spec}. The matrix $\bm Q$ is 
\begin{align*}
    \bm Q = \bm Q_1 = (\bm F_1 \otimes \bm I_2)^\ast \bm D_1 (\bm F_1 \otimes \bm I_2),
\end{align*}
where $\bm D_1 = \bigoplus_{\omega_j \in S_J} \bm\Phi_1(\omega_j)$ and $\bm I_2$ is the $2 \times 2$ identity matrix.

\subsection{Selection Criteria for \texorpdfstring{$K$}{1}, \texorpdfstring{$L_k$}{1} and \texorpdfstring{$h_{\text{max}}$}{1}}\label{parameter_selection}
In this subsection, we establish the selection criteria for $K$ and $L_k$ in the finite truncation approximation \eqref{truncate_marginal_dfpca} and $h_{\text{max}}$ for the Bartlett lag-window estimators \eqref{individual_spec_est}.

For the number of components \( K \), we adopt the variance explained ratio criterion \citep{tan2024graphical, guo2024unified}.
Specifically, we choose \( K \) by maximizing: 
\begin{align*}
    r(k) = \frac{\int_{-\pi}^{\pi} \hat{\eta}_k(\omega) \mathrm{d} \omega}{\int_{-\pi}^{\pi} \hat{\eta}_{k+1}(\omega) \mathrm{d} \omega}, \quad k = 1, \ldots, K_{\operatorname{max}},
\end{align*}
where \( K_{\operatorname{max}} \) is a predetermined upper limit, and \( \hat{\eta}_k(\omega) \) denotes the \( k \)th eigenvalue of the estimated marginal spectral density kernel \( \hat{f}_{\mathcal{S}}(\cdot, \cdot \mid \omega) \) defined in \eqref{margin_spec_est}.

The value of $L_k$ is determined such that the cumulative norm of estimated functional filters $\hat{\phi}_{kl}(\cdot)$ satisfies
\begin{align*}
    \sum_{|l| \leq L_k} \|\hat{\phi}_{kl}\|^2 \geq 1 - \epsilon, \quad k \geq 1,
\end{align*}
where $\epsilon$ denotes the tolerance threshold. This selection criterion is widely adopted in other literature on frequency-domain approaches \citep{hormann2015dynamic, tan2024graphical, guo2024unified}. Throughout this study, we choose $\epsilon = 0.1$, and set a upper bound $L_{\text{max}}$ for all $L_k$, i.e., when $\sum_{|l| \leq L_{\text{max}}} \|\hat{\phi}_{kl}\|^2 < 0.9$, we simple set $L_k = L_{\text{max}}$.

To select the truncation number $h_{\text{max}}$, we follow the rule of thumb \citep{rubin2020sparsely, guo2024unified} and set \( h_{\text{max}} = (J \overline{N})^{\frac{1}{4}} \), where \( J \) denotes the sample size and \( \overline{N} = \left( \sum_{i \in [p]} \sum_{j \in [J]} N_{ij} \right)/(pJ) \) is the average number of observations across all curves.

\

\subsection{Algorithm Details for Scores Extraction}\label{MAP_details}

To obtain the Maximum A Posteriori (MAP) estimator for the scores $\{\boldsymbol{\xi}_{\cdot \cdot k}; k \in [K]\}$, we implement a gradient ascend algorithm with respect to the log-transformed posterior distribution $\log \pi\left(\boldsymbol{\xi}_{\cdot \cdot 1}, \ldots, \boldsymbol{\xi}_{\cdot \cdot K} \mid \boldsymbol{Y}\right)$ with definition in \eqref{posterior}. The corresponding gradient is given by
\begin{align*}
 &\frac{\partial \log \pi\left(\boldsymbol{\xi}_{\cdot \cdot 1}, \ldots, \boldsymbol{\xi}_{\cdot \cdot K} \mid \boldsymbol{Y}\right)}{\partial \boldsymbol{\xi}_{\cdot \cdot k}} \nonumber \\
 &= \bigg\{ \frac{\partial \log \pi\left(\boldsymbol{\xi}_{\cdot \cdot 1}, \ldots, \boldsymbol{\xi}_{\cdot \cdot K} \mid \boldsymbol{Y}\right)}{\partial \boldsymbol{\xi}_{1 \cdot k}} , \ldots, \frac{\partial \log \pi\left(\boldsymbol{\xi}_{\cdot \cdot 1}, \ldots, \boldsymbol{\xi}_{\cdot \cdot K} \mid \boldsymbol{Y}\right)}{\partial \boldsymbol{\xi}_{p \cdot k}} \bigg\} ^{\top},
\end{align*}
where 
\begin{align*}
 &\frac{\partial \log \pi\left(\boldsymbol{\xi}_{\cdot \cdot 1}, \ldots, \boldsymbol{\xi}_{\cdot \cdot K} \mid \boldsymbol{Y}\right)}{\partial \boldsymbol{\xi}_{i \cdot k}} \nonumber \\
 &=  - {\sigma_i}^{-2}\Bigg( \sum_{k^{\prime} \in [K]} \boldsymbol{\xi}_{i\cdot k^{\prime}}\sum_{j \in [J]}{\boldsymbol{\phi}}^{*}_{i k^{\prime} j}{\boldsymbol{\phi}}_{i k j} - \sum_{j \in [J]} \tilde{\boldsymbol{Y}}_{ij}{\boldsymbol{\phi}}_{ikj}\Bigg) \nonumber  
- \Bigg[ \operatorname{Re}\bigg\{ \sum_{j \in [J]} {\Phi}_k(\omega_j) \boldsymbol{\xi}_{\cdot \cdot k}   \boldsymbol{\rho}_k\left(\omega_j\right) \boldsymbol{\rho}_k\left(\omega_j\right)^{*}\bigg\} \Bigg]_{i.},
\end{align*}
for any $i \in [p]$. Here, $\tilde{\boldsymbol{Y}}_{ij}:= \{Y_{ij1} - {\mu}_i(t_{ij1}), \ldots, Y_{ijN_{ij}} - {\mu_i}(t_{ijN_{ij}})\}$ denotes the demeaned observations, ${\boldsymbol{\phi}}_{ikj}$ is a $N_{ij} \times (J+2L_k)$ matrix with the $(z, j + l + L_k)$th element being $\phi_{kl}(t_{ijz} )$ for $|l| \leq L_k$ and 0 otherwise, and $\boldsymbol{\rho}_k\left(\omega\right):=\frac{1}{\sqrt{2\pi(J+2 L_k)}}\left\{\exp \left(\mathrm{i} 1 \omega\right), \ldots, \exp \left(\mathrm{i}\left(J+2 L_k\right) \omega\right)\right\}^{*}$. The operators $\operatorname{Re}(\cdot)$ and $[\cdot]_{i\cdot}$ extract the real part and the $i$th row of a matrix, respectively.

\

\subsection{Cross-individual Dependence Structure}\label{score_gen}
For cross-individual dependence, we introduce graphical structures for the scores $\boldsymbol{\xi}_{\cdot jk}$ \citep{zapata2022partial, tan2024graphical}. Specifically, we first sample edges between $p$ distinct subjects independently with probability $\kappa / p$, yielding the edge set $E$. Given $E$, the $(i_1, i_2)$-th element of the precision matrix $\boldsymbol{\Theta}_k^b$, i.e., the inverse covariance matrix, for the innovations $\boldsymbol{b}_{\cdot jk}$ is specified as
\begin{align*}
    \left[\boldsymbol{\Theta}_k^b\right]_{i_1, i_2} =
\begin{cases}
\frac{1}{5} \exp\left(\frac{k}{10}\right) & \text{if } i_1 = i_2, \\
R_{i_1, i_2} \exp\left(\frac{k}{10}\right)/5 & \text{if } (i_1, i_2) \in E \text{ and } i_1 \neq i_2, \\
0 & \text{if } (i_1, i_2) \notin E,
\end{cases}
\end{align*}
where $R_{i_1, i_2} \sim \text{Unif}([-r_2, -r_1] \cup [r_1, r_2])$. Based on $\boldsymbol{\Theta}_k^b$, $\boldsymbol{b}_{\cdot jk}$ are sampled from the multivariate Gaussian distribution $\mathcal{N}(\boldsymbol{0}, (\boldsymbol{\Theta}_k^b)^{-1})$. In the following, we set $\kappa = 3$ and $[r_1,r_2] = [0.1, 0.35]$.

\end{document}